\documentclass[% reprint,
nofootinbib,
 amsmath,amssymb,
 aps,
pra,
onecolumn
]{revtex4-2}

\usepackage{float}
\usepackage{algorithm}
\usepackage{algpseudocode}
\usepackage{graphicx}
\usepackage{subfigure}
\usepackage{amssymb,amsfonts,amsmath,amsthm}
\usepackage{braket}
\usepackage{hyperref}
\usepackage{color}
\usepackage{soul}
\usepackage{changes}

\usepackage{enumerate,enumitem}
\usepackage{comment}
\usepackage{hyperref}
\usepackage{eqnarray}
\usepackage{bbold}
\usepackage{optidef}
\usepackage{booktabs}

\newtheorem{theorem}{Theorem}

\newtheorem{lemma}{Lemma}

\def\D{\mathcal{D}}

\def\H{\mathcal{H}}

\def\S{\mathcal{S}}

\def\afiliacion{\textit{Facultad de Matem\'atica, Astronom\'ia, F\'isica y Computaci\'on, Universidad Nacional de C\'{o}rdoba, C\'ordoba and CONICET, Argentina}}

\usepackage{xcolor}
\begin{document}

\title{Quantum imaginary time evolution and UD-MIS problem}%  

\author{Victor A. Penas}
\email{vapenas@unrn.edu.ar}

\affiliation{\afiliacion}

\author{Marcelo Losada}
\affiliation{\afiliacion}

\author{Pedro W. Lamberti}
\affiliation{\afiliacion}

\date{\today}

\begin{abstract}
In this work we apply a procedure based on the quantum imaginary time evolution method to solve the unit-disk maximum independent set problem.
Numerical simulations are performed for instances of 6-, 8- and 10-qubit graphs. We find that the failure probability of the procedure is relatively small and rapidly decreases with the number of shots. In addition, a theoretical upper bound for the failure probability of the procedure is obtained.

\end{abstract}

\maketitle

\section{Introduction}

Obtaining the  Hamiltonian ground state of a quantum system is of utmost importance for various physics problems and for optimization problems.
Quantum computing may be useful for calculating these states. 
Several quantum algorithms have been proposed for this purpose, including adiabatic quantum optimization \cite{Farhi_2001}, quantum annealing \cite{Das_2008}, and classical quantum variational algorithms such as the quantum approximate optimization algorithm (QAOA) \cite{farhi2014quantumapproximateoptimizationalgorithm, farhi2014quantumapproximateoptimizationalgorithm,Zhou_2020,ostrowski2020lowerboundscircuitdepth,Shaydulin_2021,Medvidovi__2021,PhysRevResearch.4.033029,herrman2021multianglequantumapproximateoptimization}
and the variational quantum eigensolver (VQE) \cite{Peruzzo_2014, Cerezo_2021, Tilly_2022}. Despite many advances, these algorithms also have potential drawbacks, especially in the context of near-term quantum computing architectures with limited quantum resources. For example,
adiabatic quantum optimization may only provide polynomial speedups \cite{Lucas:2013ahy, Abbas_2024}, while variational quantum algorithms are limited in accuracy by a fixed ansatz and involve high-dimensional classical optimizations \cite{McClean_2018}.

Another approach to finding ground states,
based on  quantum imaginary time evolution (QITE), has been proposed by Motta et al. in \cite{Motta_2019}. This approach emulates  classical imaginary time evolution with measurement-assisted unitary circuits. Evolution in imaginary time evolves the system to zero temperature \cite{Love_2020}, returning the ground state. QITE has found practical application in quantum chemistry on 
noisy intermediate-scale quantum (NISQ) hardware \cite{Barison_2022, Yeter-Aydeniz_2020,Tsuchimochi_2023}, in simulating open quantum systems \cite{Kamakari_2022}, thermodynamic observables \cite{Getelina_2023,Gacon_2024} and recently in optimization problems such as Max-Cut and 
polynomial unconstrained Boolean optimization (PUBO) \cite{alam2023solvingmaxcutquantumimaginary,PhysRevA.109.052430}. Similar to the QAOA, QITE involves variational parameters to be optimized, but unlike the QAOA, the parameters are fixed by algebraic equations \cite{alam2023solvingmaxcutquantumimaginary}. QITE seems to have an advantage \cite{Motta_2019} over variational quantum algorithms since usually these methods require high-dimensional classical optimizations.

The MaxCut and the Maximum Independent Set problems are prototypical examples of optimization problems that have received attention as candidates for quantum advantage \cite{Guerreschi_2019,otterbach2017unsupervisedmachinelearninghybrid,pichler2018quantumoptimizationmaximumindependent,Henriet_2020,Ebadi_2022,PhysRevResearch.5.043277,Lykov_2023,Yeo_2024,PhysRevA.108.052423,Henriet_2020_review,Kim2023}. The unit-disk maximum set problem (UD-MIS) is computationally challenging. It belongs to the class of NP-hard problems, which means that finding an exact solution efficiently for large instances is infeasible. 
Research on quantum hardware and quantum algorithms has been addressed in \cite{Labuhn_2016, Barredo_2018, Serret_2020, pichler2018quantumoptimizationmaximumindependent} to tackle UD-MIS problems. For instance, in \cite{Serret_2020} quantitative requirements on system sizes and noise levels of Rydberg atoms platforms were studied to reach quantum advantage in solving UD-MIS problem with a quantum annealing-based algorithm.

In this work, we apply a procedure based on QITE algorithm proposed in \cite{Motta_2019} to solve UD-MIS problems. We simulate, without noise, QITE with two non-trivial domains and apply it to several UD-MIS instances of 6, 8 and 10 qubits. We explore the number of iterations and shots required for our probabilistic procedure to return acceptable results.

The paper is organized as follows. In Section \ref{sec:qite} we review the quantum imaginary time evolution method presented in \cite{Motta_2019}. We provide some notations and definitions that will be useful later on. 
In Section \ref{sec:optimization-based-on-qite} we describe our proposed method for solving  optimization problems using QITE. 
In Section \ref{sec:udmis} we briefly  describe the UD-MIS problem and we explain how this problem can be formulated in terms of the Hamiltonian
of a quantum system.
In Section \ref{sec:numeric-results} we numerically test
the proposed method on several instances of UD-MIS problems.
In Section \ref{sec:conclusions}, we present some concluding remarks. For readability, an auxiliary lemma, proofs and some explanations of numerical calculations are presented separately in Appendices~\ref{app:aux_proofs}  and~\ref{app:numerical}, respectively.

\section{Quantum imaginary time evolution method}
\label{sec:qite}

In \cite{Motta_2019}, Mota et al. proposed an algorithm for finding ground states, based on  quantum imaginary time evolution. This approach emulates  classical imaginary time evolution 
with measurement-assisted unitary circuits. Some improvements to the QITE algorithm have been developed for NISQ hardware \cite{Nishi_2021,Ville_2022,Cao_2022, Benedetti_2021,Gacon_2023,Gacon_2024}.  
Other approaches to QITE have been proposed. For instance, Monte Carlo quantum simulations \cite{Huo_2023}, techniques with reinforcement learning \cite{Cao_2022}, probabilistic methods of implementing non-unitary operations \cite{Liu_2021,Lin_2021,PhysRevResearch.4.033121}, alternative approaches where an orthogonal basis set is used in each propagated state \cite{Jouzdani_2022},  and variational versions  \cite{McArdle_2019, Beach_2019,Gacon_2023,Gacon_2024,Kolotouros_2025}. In \cite{Kolotouros_2025} the authors improved upon a variational version of quantum imaginary time evolution \cite{McArdle_2019} by estimating quantum Fisher information matrix using random measurements and considering the average classical Fisher information matrix. 
In this section we review the quantum imaginary time evolution method proposed in \cite{Motta_2019}, which is the method used throughout this work. 

The imaginary time evolution method is a well-known approach used for obtaining the ground state of a quantum Hamiltonian. 
 The idea is to express the ground-state $\ket{\psi}$ of the Hamiltonian as the long-time limit of the imaginary time Schr\"odinger equation, that is, 
 
 \begin{equation}\label{eq:generic-ite}
 \ket{\psi} = \lim_{t \rightarrow \infty} \ket{\psi^{ite}_t}, ~~~ \ket{\psi^{ite}_t} = \frac{ e^{-tH}\ket{\psi_0} } { ||  e^{-tH}\ket{\psi_0}  || },
 \end{equation}
 for some initial state $|\psi_0\rangle$, with the condition $\langle \psi_0|\psi\rangle \neq 0$. The parameter $t>0$ must be chosen such that the final state is close enough to the ground state.
This parameter can also be thought of as the inverse temperature ($t = \frac{1}{k_B T}$, where $k_B$ is the Boltzmann constant and $T$ is the temperature). In this case, Eq. \eqref{eq:generic-ite} states that the system tends to the ground state when temperature tends to zero. 
 
In \cite{Motta_2019} it was proposed a quantum imaginary time evolution method (QITE) that emulates Eq. \eqref{eq:generic-ite}. It consists of measurement-assisted unitary circuits, acting on suitable domains around the support of different qubits. To be more precise, start with a geometric $k$-local Hamiltonian with $m$ terms of the form:
\begin{equation}\label{eq:sum-hm}
    H = \sum^{m}_{l= 1} h[l].
\end{equation}
Each term $h[l]$ acts on at most $k$ (neighboring) qubits on an underlying graph. 
Let us consider 
a Trotter decomposition \cite{Trotter_1959,Suzuki_1991} of the corresponding imaginary
time evolution 
\begin{equation}\label{eq:trotter-decomposition}
    e^{-tH} =  (e^{- \tau h[1]}e^{- \tau h[2]}\cdots e^{- \tau h[m]})^n + \mathcal{O}(\tau^2), ~~~ \tau =t/ n,
\end{equation}
acting on an initial state $\ket{\psi_0}$ (usually taken to be a product state). In the decomposition we have $n$ iteration of the form $e^{- \tau h[1]}e^{- \tau h[2]}\cdots e^{- \tau h[m]}$, and each $e^{-\tau h[l]}$ is a sub-step of a complete iteration step. The interval $\tau > 0$ must be chosen such that the squared errors are negligible. 

After a Trotter sub-step $e^{-\tau h[l]}$ of the iteration $j$ we have
\begin{equation}
    |\psi'\rangle=\frac{e^{-\tau h[l]}}{||e^{-\tau h[l]}|\psi\rangle ||}|\psi\rangle.
\end{equation}

The idea is to map the scaled non-unitary action $e^{-\tau h[l]}$ of the iteration $j$ on the state $\ket{\psi}$ to that of a unitary evolution $e^{-i\tau A[l,j]}$:
\begin{equation}\label{eq:qite-idea}
    |\psi'\rangle \approx e^{-i\tau A[l,j]}|\psi\rangle.
\end{equation}
Here $A[l,j]$ is a Hermitian operator associated with the $l$-th term $h[l]$ of iteration $j$ acting on a qubit domain $D_{l,j}$ of size  $d_{l,j}$, where $D_{l,j}$ is the set of qubits where $A[l,j]$ acts non-trivially. The domain is usually chosen around the support of $h[l]$.  $A[l,j]$ can be expanded as a sum of Pauli strings acting on $D_{l,j}$, 
\begin{equation}\label{eq:hermitian-A}
    A[l,j]=\sum_Ia[l,j]_I\sigma_I,~~I=i_1 i_2\ldots i_{d_{l,j}},
\end{equation}
with $\sigma_I=\sigma_{i_1}\otimes\cdots\otimes\sigma_{i_{d_{l,j}}}$, and    $\sigma_{i_k}$ being a Pauli matrix with $i_k \in \{ I, X, Y, Z \}$, and it is understood that identity matrices should be inserted in the product when appropriate. 

The goal is to minimize the difference $|| \ket{\psi'}-(1-i\tau A[l,j])\ket{\psi} ||^2$ with respect to real variations of $a[l,j]_I$, where $|| \cdot  ||$ is the norm of the  Hilbert space. Up to $\mathcal{O}(\tau^2)$ errors, this difference translates into a linear problem in the coefficients $a[l,j]_I$:
\begin{equation}
\label{eq:linear_system}
    \sum_{J}(S + S^T)_{IJ} a_J[l,j] = -b_I,
\end{equation}
where $S_{IJ}=\langle \psi|\sigma_I\sigma_J |\psi \rangle$ and $b_I=-2\mbox{Im}\langle \psi|\sigma_Ih[l] |\psi \rangle$. 
In order to obtain $S_{IJ}$ and $b_I$, we need to perform suitable measurements of the state $|\psi \rangle$.
Solving the linear equation \eqref{eq:linear_system}, we obtain the coefficients $a[l,j]$. From $a[l,j]$ we obtain the unitary evolution $e^{-i\tau A[l,j]}$ for the corresponding sub-step. 
Finally, iterating this process for the $m$ terms of the Hamiltonian and the $n$ iteration steps, 
we construct the QITE operator:
\begin{equation}
 \mbox{Q}_H(\tau, n, \mathcal{D})= \prod_{j=1}^{n}\prod_{ l=1}^{ m} e^{-i\tau A[l,n +1-j]}.
\end{equation}
Here $\mathcal{D}$ contains the information of all domains $D_{l,j}$, with $1\le l \le n$ and $1\le j \le m$.
Finally, the QITE evolution is given by
\begin{equation}
|\psi_0\rangle\rightarrow \ket {\phi^{qite}_t } = \mbox{Q}_H(\tau, n, \mathcal{D})|\psi_0\rangle.
\end{equation}
For simplicity, in the Section \ref{sec:numeric-results} we will fix all domains to be equal on each step (but not on each sub-step), that is, for each $1\leq l \leq m$, 
$D_{l,1} =\cdots= D_{l,n}$.

In constructing the  QITE operator we have to define for each sub-step the domain of qubits where the Pauli strings act on.
In \cite{Motta_2019} it is argued that as the number of sub-steps increase, the size of the domain should increase from an initial domain that involves the qubits associated with the natural support of each $h[l]$. 
In general, this is due to an expected increase in correlations between qubits when starting from a product state. However, it is argued that for systems with finite correlations over at most $C$ qubits bounded by $\mbox{exp}(-L/C)$, with $L$ the Manhattan distance between the observables on the lattice, the domain 
can be chosen with a width of at most $\mathcal{O}(C^D)$ surrounding the qubits on the support of $h[l]$ ($D$ the dimension of the underlying regular lattice). They also argue that the number of measurements and classical storage at a given time step required to perform the QITE computation is bounded by $\mathcal{O}(\mbox{exp}(C^D))$, but it scales in a quasi-polynomial way in terms of the numbers of qubits. 
In practice, however, we can choose a fixed domain smaller than the one induced by $C$ by truncating the unitary updates on each step to domain sizes that fit the computational budget. Of course, the larger the size the better the approximation to the ground-state. Nevertheless, this approximate version of QITE remains useful as a valid heuristic and we use it in this work.

\section{Optimization method based on QITE}\label{sec:optimization-based-on-qite}

Let us consider a quantum system of $N$ qubits, with Hilbert space $\mathcal{H} = \mathbb{C}^2 \otimes \cdots \otimes \mathbb{C}^2$, and dimension $d = 2^N$. We consider a diagonal Hamiltonian in the computational basis $\{ |i\rangle \}_{0 \leq i \leq d-1}$, given by 
\begin{equation}
\label{eq:H_optimization}
H = a I +  \sum_{i= 0}^{N-1} b_i Z_i + \sum_{(i, i')}^{N-1} c_{i,i'} Z_i Z_{i'}, 
\end{equation}
with $I$ the identity operator, $Z_i$ the $Z$ Pauli matrix acting on the i-th qubit, and $a$, $b_i$, $c_{i,i'}$ real coefficients.
We denote $E_0, \ldots , E_{d-1}$ the eigenvalues of the Hamiltonian $H$, with all the $E_i$ sorted in an non-decreasing way (each eigenvalue has to be considered with their respective degeneracy). We denote the corresponding eigenstates respectively as $| E_0 \rangle, \ldots, | E_{d-1} \rangle$, where some $|E_i\rangle$ should be understood as belonging to the same eigenspace. 
Since the Hamiltonians are
diagonals in the computational basis, that is, 
the eigenvectors of $H$ are the vectors of the computational basis, the basis $\{|E_i\rangle \}_{0 \leq i \leq d-1}$ can be chosen such that it has the same elements of the computational basis, but in different order. In this work, the basis $\{|E_i\rangle \}_{0 \leq i \leq d-1}$ is chosen in such a way.

The problem we want to solve is the following: 
\textit{Finding the lowest eigenvalue $E_0$ of $H$ or an eigenvalue 
$E_i$ such that $E_i \leq E_0 + \delta E$, with $\delta E \geq 0$ a tolerable error.}

In order to solve this problem, we propose the following probabilistic method (see Alg.\ref{alg:cap}) based on QITE:

\begin{enumerate}
    \item We start with an initial state  $|\psi_0\rangle =\mathbf{H}^{\otimes N}|0 \ldots 0\rangle$, where $\mathbf{H}^{\otimes N}$ is the Hadamard gate acting on each qubit. In terms of the eigenvectors basis of $H$, the initial state has the form $\ket{\psi_0} = \frac{1}{\sqrt{d}}\sum_{i = 0}^{d-1} \ket{E_i}$.

    \item Then, we apply the QITE operator up to a time $t_{max}$:
\begin{equation}\label{eq:qite-operator}
    |\psi_0\rangle\rightarrow \ket {\phi^{qite}_{t_{max}} } =  \mbox{Q}_H(\tau, n_{max}, \mathcal{D})|\psi_0\rangle,
\end{equation}
where $t_{max}$, $\tau$, and $\mathcal{D}$ should be chosen according to the particular problem, and $n_{max} = t_{max} /\tau$.
    
 \item We measure $|\phi^{qite}_t\rangle$ $M$ times ($M$ number of shots) in the computational basis, with  $M << d$. We register the $M$ outputs $E_{i_m}$, with $1 \leq m \leq M$.

 \item Finally, with a classical computer, we choose from the
 $M$ outputs $E_{i_m}$ the one with less energy.  
 
\end{enumerate}

The general idea of the method is the following.
When applying QITE operator $\mbox{Q}_H(\tau, n_{max}, \mathcal{D})$ to the initial state $\ket{\psi_0}$, as $n_{max}$ grows, it is expected to obtain a superposition of eigenstates of $H$ with higher probabilities for the eigenstates with the smallest eigenvalues. 
In the case of Hamiltonians of the form given in Eq. \eqref{eq:H_optimization}, the Hamiltonian basis coincides with the computational basis. 
Therefore, when measuring in the computational basis, we obtain eigenstates of $H$, and we expect to obtain with more probability the eigenstates with lower eigenvalues.
If we consider $M$ shots,  the more we increase $M$, the more likely we are to get a good result, that is, a state with eigenvalue $E_i$ such that $E_i \leq E_0 + \delta E$ with $\delta E$ a tolerable error.  However, the number of shots should not increase exponentially with the number of qubits, otherwise  in the fourth step of the procedure we end up in a $for$ loop with an exponential number of iterations. 
In Section \ref{sec:numeric-results}, we  numerically show for 6-, 8- and 10-qubits graphs that the number of shots $M$ do not need to be drastically increased  with the number of qubits.

It should be noted that the state obtained by our method and the QITE state $|\phi^{qite}_t\rangle$ are not necessarily the same. While the first one is an eigenstate of $H$, obtained after a measurement, the second one, in general, is a superposition of eigenstates.   When this output state is such that $E_i \leq E_0 + \delta E$ with $\delta E$ a tolerable error, we call it an \textit{acceptable state}.

Since the proposed method is probabilistic, there is a  failure probability associated with it, and the lower it is, the higher the chances of obtaining at least one acceptable eigenvalue of $H$ after $M$ shots.
Given the QITE state at time $t$,  $\ket{\phi^{qite}_t}$, when measuring one time ($M=1$) in the computational basis, the \textit{failure probability} is the probability of obtaining an eigenvalue $E_i > E_0+\delta E$, which is given by 
\begin{equation}\label{eq:qite-failure-probality}
        P_F^{qite}(t)= \sum_{E_i > E_0+\delta E }
        | \langle E_i \ket{\phi^{qite}_t} | ^2.   
\end{equation}
When measuring $M$ times in the computational basis, the failure probability for the QITE state at time $t$ is $(P_F^{qite}(t))^M$.

In the ideal case, the QITE state at time $t$, $\ket{\phi^{qite}_t}$, closely approximates the ITE state at time $t$, $\ket{\psi^{ite}_t}$, given by Eq. \eqref{eq:generic-ite}. The ITE state sometimes will be called the exact state in Section \ref{sec:numeric-results}. It can be expressed in terms of the eigenvectors of the Hamiltonian as follows:
\begin{equation}
\ket{\psi^{ite}_t}= \frac{1}{\mathcal{K}}\sum_{i = 0}^{d-1} e^{-t E_i} | E_i\rangle, ~~~ \mbox{with} ~~~ \mathcal{K}^2  = \sum_{i = 0}^{d-1} e^{-2 t E_i}.
\end{equation}
We can compute the failure probability for the ITE state: 
\begin{equation}\label{eq:ite-failure-probality}
        P_F^{ite }(t)=\sum_{E_i > E_0+\delta E}
        | \langle E_i \ket{\psi^{ite}_t} | ^2 = \sum_{E_i > E_0+\delta E  } \frac{e^{- 2 t E_i}}{\mathcal{K}^2}.
\end{equation}

We are interested in an upper bound for the failure probability independent of the Hamiltonian $H$. The following theorem provides a relevant bound for the ITE failure probability.

\begin{theorem}[\textbf{ITE failure probability upper bound}]\label{th: cota_ite}
    Let $H = \sum_{i= 0}^{d-1} E_i \ket{E_i}\bra{E_i}$ be a Hamiltonian of a quantum system with the fundamental eigenvalue 
    with degeneracy $g$,  $\delta E \geq 0$ a tolerance, and $t > 0$ a time value. The failure probability for $\ket{\psi^{ite}_t}$ satisfies the inequality
    \begin{equation} \label{eq:ite_bound}
    P_F^{ite}(t) \leq \frac{ 1}{ 1 +g (d-g)^{-1}e^{2t\delta E}  }.  
    \end{equation}    
    
\end{theorem}

The proof of this theorem is presented in Appendix \ref{app:aux_proofs}. From this result we can obtain useful relations between the parameters of the proposed method. Eq. \eqref{eq:ite_bound} can be restated as $  P_F^{ite}(t) \leq 1/( 1 + g e^{2t\delta E - N \ln{2}})$.
When measuring $M$ times,  this expression has an exponent $M$. Then, it is enough that the parameters satisfy $2t\delta E \gtrapprox  N \ln{2}$ in order to obtain a good upper bound. 
This implies that time need not increase more than linearly with the number of qubits.

The ITE upper bound can be connected with the QITE upper bound. The following theorem provides a relation between the ITE failure probability and the QITE failure probability.

\begin{theorem}[\textbf{QITE failure probability upper bound}]\label{th: cota_qite}
    
     Let $H = \sum_{i= 0}^{d-1} E_i \ket{E_i}\bra{E_i}$ be a Hamiltonian of a quantum system, 
     $\delta E \geq 0$ a tolerance, and $ \ket {\psi^{ite}_t}$ and $ \ket {\phi^{qite}_t}$ the corresponding ITE and QITE states at time $t>0$. 
     Given $0 \leq  \varepsilon \leq \sqrt{2}$, if $||\ket {\psi^{ite}_t} -  \ket{\phi^{qite}_t}|| \leq \varepsilon$, the failure probabilities $P_F^{ite}(t)$ and $P_F^{qite}(t)$ satisfy 
    \begin{equation}\label{eq:qite-upper-bound}
        |P_F^{qite}(t)- P_F^{ite}(t)| \leq   \varepsilon \sqrt{1 - \frac{\varepsilon^2}{4} }.
    \end{equation}

\end{theorem}

The proof of this theorem is presented in Appendix \ref{app:aux_proofs}. Motivated by this theorem we define the following quantities that will be analyzed later on:

\begin{equation}\label{eq:varepsilon}
    \varepsilon(t)=||\ket {\psi^{ite}_t} -  \ket{\phi^{qite}_t}||,
\end{equation}
\begin{equation}\label{eq:Barepsilon}
\Bar{\varepsilon}(t) = ||\ket {\psi^{ite}_{t_{max}}} -  \ket{\chi_t}||,
\end{equation}
where $\ket{\chi_t}$ could be either $\ket{\psi^{ite}_{t}}$ or $\ket{\phi^{qite}_{t}}$. The quantity $\varepsilon (t)$ tells us how much the QITE state departs from the ITE state at time $t$, and  $\Bar{\varepsilon}(t)$ tells how much the corresponding state at time $t$ departs from the ITE state at a final time $t_{max}$. For an iterative process, $\Bar{\varepsilon}(t)$ shows how well $\ket{\chi_t}$ is approximating the expected  solution $\ket {\psi^{ite}_{t_{max}}}$ at a certain iteration step.

Thm. \ref{th: cota_qite} provides an interesting upper bound for the QITE failure probability when $\varepsilon(t) \leq \sqrt{2}$:
\begin{equation}  P_F^{qite}(t)   \leq P_F^{ite}(t) +    \varepsilon \sqrt{1 - \frac{\varepsilon^2}{4} }.
\end{equation}
Moreover, combining with ITE failure probability upper bound, we 
obtain the following inequality: 
\begin{equation}
P_F^{qite}(t)   \leq \frac{ 1}{ 1 +g (d-g)^{-1}e^{2t\delta E}  } +    \varepsilon \sqrt{1 - \frac{\varepsilon^2}{4} }.
\end{equation}
The distance $\varepsilon(t) \leq \sqrt{2}$ depends on several factors, such as the parameter $t$, the chosen domain $\D$ and the length of the step $\tau$. In Section \ref{sec:numeric-results} and Appendix \ref{app:numerical}, we analyze how distance $\varepsilon(t)$ behaves for different numbers of qubits, iterations, and domains.

\begin{algorithm}[H]
\begin{algorithmic}
\caption{}\label{alg:cap}
\State $|\phi_0\rangle\gets \mathbf{H}^{\otimes N}|0\ldots0\rangle$ \Comment{Initial state}
\State $t \gets t_{max}$\Comment{$t_{max}=\tau \times n_{max}$, with $n_{max}$ iterations.}
\State $H \gets \mbox{Optimization Hamiltonian}$
%\State $M \gets M_0$ 
%\Function{Cost Function}{c}
%\State $H \gets a$
%\State $ansatz \gets \mbox{circuit}$
%\State \Return $costs$
%\EndFunction
\State $|\phi^{qite}_{t_{max}}\rangle \gets \mbox{Q}_H(\tau, n_{max}, \mathcal{D})|\phi_0\rangle$
%\For{$m \le M$}\Comment{$M$: Number of shots}
%\State $|\mbox{C.S.}\rangle_m\gets\mbox{Measure } |\phi^{qite}_t\rangle$
%\State $m\gets m+1$
%\EndFor
\State $\mbox{Measure } |\phi^{qite}_{t_{max}}\rangle~ 
  M \mbox{ times} ~(\mbox{computational basis}).~~\mbox{Keep each measured state:}~~|E_{i_m}\rangle,~~1\le m \le M.$\Comment{$M$: Number of shots}
\For{$m \le M$}
\State $E_{i_m} \gets \langle E_{i_m}| H|E_{i_m}\rangle$
\State $m\gets m+1$
\EndFor
\State\Return $|E_{i_m}\rangle~\mbox{associated with lowest}~E_{i_m},~1\le m \le M.$
\end{algorithmic}
\end{algorithm}

\section{Unit-disk maximum independent set problem}
\label{sec:udmis}

In this section, we describe the unit-disk maximum independent set problem (UD-MIS), a basic graph optimization problem with many applications. 
We also explain how this problem can be formulated in  terms of the Hamiltonian of a quantum system.

Let $G = (V,E)$ be a graph with vertex set $V$ and edge set $E$, and let $N$ be the number of vertices of the graph $G$. 
An independent set of $G$ is a
set of mutually non-connected vertices. 
Let $S = (s_1 , \ldots ,  s_N)$ be a bitstring of length $N$ ($s_i \in \{0,1\}$), and let $\mathcal{B}_N$ be the set of all possible bitstrings of length $N$. The size of $\mathcal{B}_N$ is exponential in
the graph size, $|\mathcal{B}| = 2^N$. 
The Hamming weight of a bitstring $S$ is given by $|S| = \sum_{i = 1}^N s_i$.
The maximum independent set (MIS) problem consists
in determining the size of the largest possible independent set and returning an example of such a set. This problem can be formulated as the following maximization problem:
\begin{equation}
\begin{aligned}
\max_{S \in \mathcal{B}} \quad & |S|\\
\textrm{s.t.} \quad & S \in I.S, 
\end{aligned}
\end{equation}
where $I.S$ (for “Independent Sets”) is the set of bitstrings
$(s_1, \ldots , s_N )$ corresponding to independent sets of $G$. A bitstring $S = (s_1,\ldots,s_N)$ corresponds to an independent set 
if for all pairs of vertices $(i, i')$ we have $s_i = s_{i'} = 1 \implies (i, i')  \notin E$. 

The UD-MIS problem is the MIS problem restricted to unit-disk graphs. A graph is a unit-disk
graph if one can associate a position in the two-dimensional plane with
every vertex such that two vertices share an edge if and
only if their distance is smaller than unity. 

This optimization  problem can be reformulated 
as a quantum minimization problem that consists in finding the ground state of a Hamiltonian of a quantum system. 
The idea is to associate each bitstring $S = (s_1 , \ldots, s_N)$ with a quantum state  $|s_1 , \ldots, s_N \rangle$ of the $N$-qubit system.
The associated Hamiltonian is the following:
\begin{equation}\label{eq:UD-MIS_Hamiltonian}
    H =  -\sum_{i \in V} \hat{n}_i + u\sum_{(i,i') \in E} \hat{n}_i  \hat{n}_{i'},
\end{equation}
with $\hat{n}_i = (I - Z_i$)/2 and $Z_i$ the Pauli matrix in the $z$ direction acting on the qubit $i$, and $u$ a parameter whose value can be adjusted. Fixing $u > 1$ guarantees that the ground state of $H$ will necessarily be an independent set.

The Maximum Independent Set problem and the MaxCut problem, are prototypical examples that have received attention as candidates for quantum advantage 
\cite{Guerreschi_2019,otterbach2017unsupervisedmachinelearninghybrid,pichler2018quantumoptimizationmaximumindependent,Henriet_2020,Ebadi_2022,PhysRevResearch.5.043277,Lykov_2023,Yeo_2024,PhysRevA.108.052423,Henriet_2020_review,Kim2023, Serret_2020}.
The unit-disk maximum set problem is computationally challenging.
It belongs to the class of NP-hard problems, which means that finding an exact solution efficiently for large instances is infeasible and that any NP optimization problem can be reformulated as a UD-MIS problem with polynomial overhead \cite{Hromkovic2004Book}.
Researchers have developed several heuristic algorithms to approximate this problem. The are two main strategies that can be distinguished. One is based on two-level shifting schemes \cite{10.1007/978-3-540-46515-7_16,DAS202063,10.1007/11604686_31,doi:10.1142/S0218195917500078} and the other is based on breadth-first-search schemes \cite{10.1007/978-3-540-30559-0_18}.
 
Here, we are going to solve it by applying the optimization method based on QITE described in Section \ref{sec:optimization-based-on-qite}.
It should be noted that the Hamiltonian given in  \eqref{eq:UD-MIS_Hamiltonian} does not depend on the actual distances of the nodes, but the links between them. This implies that correlations between nodes are affected by the links they share rather than the geometric distance between them.
The QITE method is expected to work well under the assumption of finite correlation length with a geometric k-local Hamiltonian \cite{Motta_2019}. The UD-MIS graphs do not have, in general, a regular lattice shape, as can be seen in Fig. \ref{UD-MIS-graphs}. This means that correlation lengths may not follow an exponential decay law between node distances in the graph\footnote{See \cite{Serret_2020} for a discussion of the correlation length being roughly independent of the graph
size with exponential decay in UD-MIS problem.}.   
Regardless of this, we have used QITE as a valid heuristic method.

\section{Numerical results}\label{sec:numeric-results}

\begin{figure}
        \centering
        \begin{minipage}{0.40\textwidth}
            \centering
            \subfigure[]{
            \includegraphics[width=\linewidth]{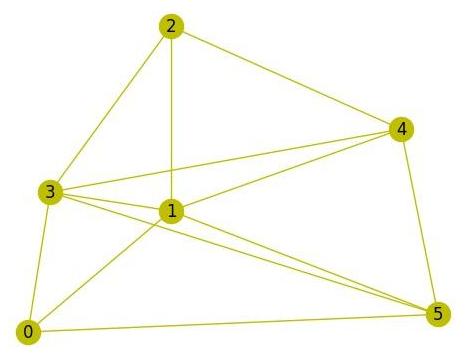}
                                    }
        \end{minipage}
        \hfill
        \begin{minipage}{0.40\textwidth}
            \centering
            \subfigure[]{
            \includegraphics[width=\linewidth]{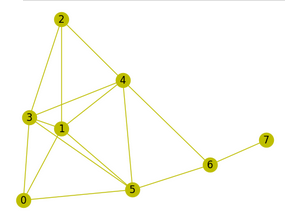}
                                    }
        \end{minipage}
        \hfill
        \begin{minipage}{0.40\textwidth}
            \centering
            \subfigure[]{
            \includegraphics[width=\linewidth]{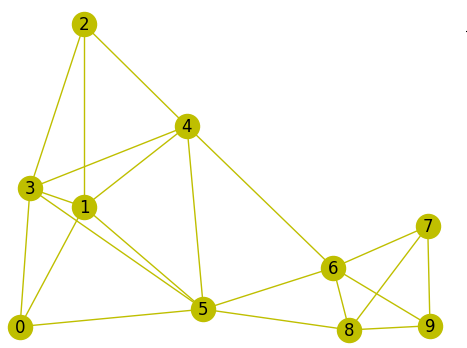}
                                    }
        \end{minipage}\caption{Examples of UD-MIS graphs consisting of 6 (a), 8 (b) and 10 (c) qubits.}
    \label{UD-MIS-graphs}
    \end{figure}

In this section we apply the optimization method based on QITE (Alg. \ref{alg:cap}) described in Section \ref{sec:optimization-based-on-qite}, to solve the UD-MIS problem presented in  Section \ref{sec:udmis}. 
In this section we summarize our numerical findings. We first explore several numerical quantities and analyze the failure probability for graphs of Fig. \ref{UD-MIS-graphs}. Then, we test the method by sampling several random graphs.
We use Hamiltonians of the form given in Eq. \eqref{eq:UD-MIS_Hamiltonian} with $u=1.35$. 
Finally, we compare QITE algorithm with quantum adiabatic computation and quantum annealing.

\subsection{Failure probability characterization}
\label{sec: Failure probability characterization}

We recall from section \ref{sec:optimization-based-on-qite} that the aim of our method is to obtain an \textit{acceptable state}, that is, 
an eigenstate with eigenvalue 
$E_i$ such that $E_i \leq E_0 + \delta E$, with $\delta E \geq 0$ a tolerable error.
In what follows, we numerically explore several quantities for the graphs instances depicted in Fig. \ref{UD-MIS-graphs} in order to understand how the performance of the method relates to the number of iterations, the domains, and the number of shots. The graphs were constructed using a square grid of size $2.2 \times 2.2$ and sampling two-dimensional points using a uniform distribution.
The Hamiltonians for these graphs are constructed according to Eq. \eqref{eq:UD-MIS_Hamiltonian}.The results will be used later on when testing Alg. \ref{alg:cap} on several instances of random graphs.

In an ideal case, that is when we choose a full domain $\mathcal{D}$ and a very small interval $\tau$,
the QITE state at time $t$, $\ket{ \phi^{qite}_t}$, should closely approximate the ITE state at time $t$, $\ket{\psi^{ite}_t}$.
In practical cases, the matching between these two states would depend on the chosen domain and the number of iterations. The difference between them can be characterized by $ \varepsilon(t)$, given in Eq. \eqref{eq:Barepsilon}, or by the fidelity.
The closer the QITE state is to the ITE state, the closer the QITE failure probability  $P_F^{qite}(t)$ is to the ITE failure probability   $P_F^{ite}(t)$, and the better the results of our algorithm should be. 
Thm. \ref{th: cota_qite} gives an upper bound for the difference between QITE and ITE failure probabilities in terms of $\varepsilon(t)$.
We thus start plotting $\varepsilon(t)$ and $\Bar{\varepsilon}(t)$ (Eq. \eqref{eq:varepsilon} and \eqref{eq:Barepsilon}, respectively) 
for the 6-qubit graph instance of Fig. \ref{UD-MIS-graphs}a. 
In all numerical simulations we have set $\tau=0.01$, so $t_{max}$ is fixed with the number of iterations, $n_{max}$. For the 6-qubit graph instance we use $t_{max}=10$ (i.e. $n_{max}=1000$). We choose two different domains: $\mathcal{D}_A$ and $\mathcal{D}_B$. $\mathcal{D}_A$ is chosen such that each $D_{l,j}$ equals the same qubit support of $h[l]$ for all iterations $1\le j \le n_{max}$. $\mathcal{D}_B$ is chosen in a similar way to $\mathcal{D}_A$ except that domains which are assigned to terms containing $Z_iZ_{i'}$ are expanded to contain a 4-qubit support around the linked qubits $i$ and $i'$ (see Appendix \ref{app:numerical} for details).
It should be noted that these domains imply that the evolution given by Eq. \eqref{eq:qite-operator} produces entanglement between qubits. 
This differs from other proposals where the focus is put on linear ansatz for the chosen domains \cite{alam2023solvingmaxcutquantumimaginary,PhysRevA.109.052430} (although a quadratic ansatz is also considered in \cite{PhysRevA.109.052430}). Our domains $\mathcal{D}_A$ and $\mathcal{D}_B$ are quadratic and quartic in terms of Pauli strings, respectively.

\begin{figure}
        \centering
        \includegraphics[width=0.72\textwidth]{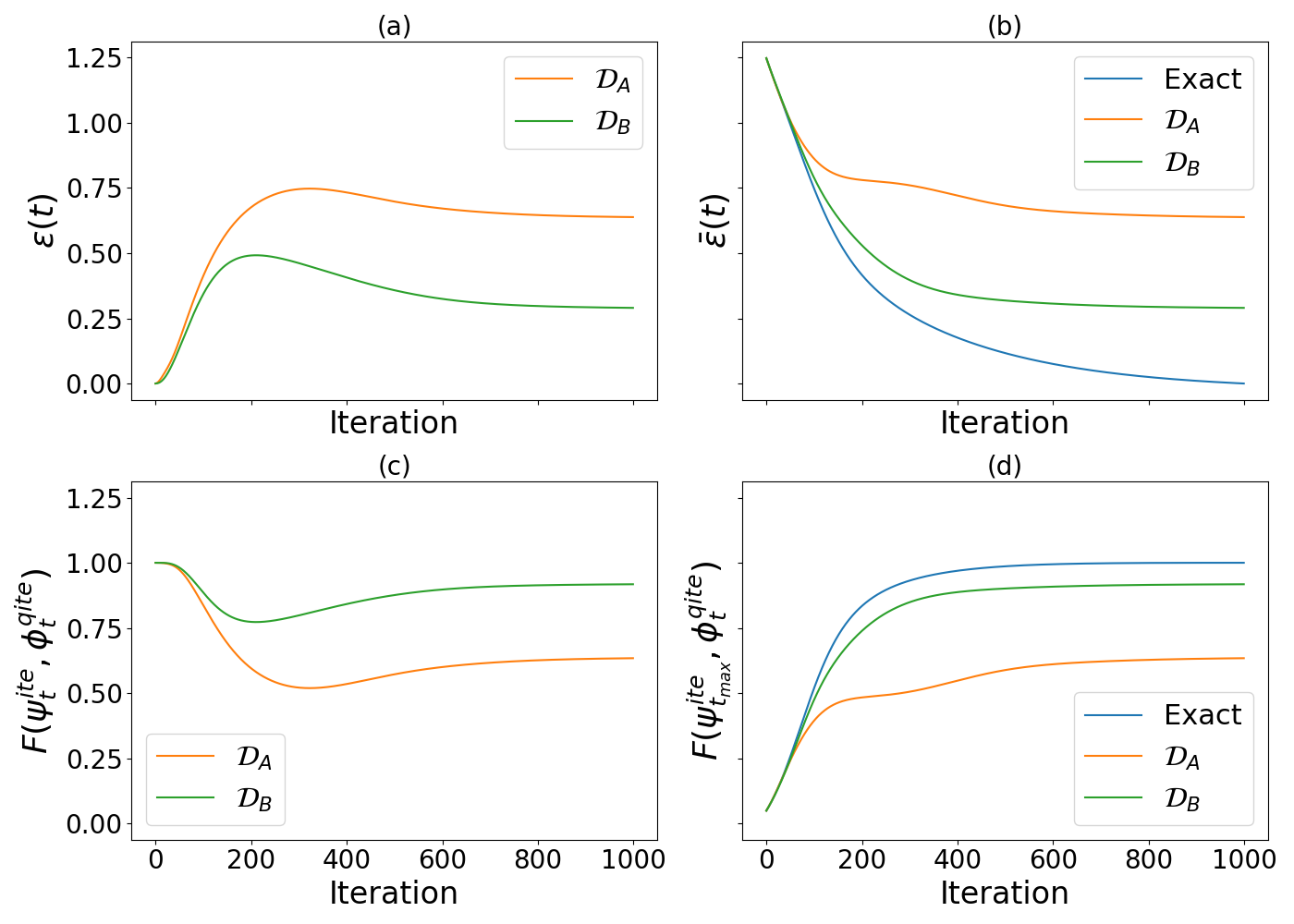} 
    \caption{Error and fidelity results for the 6-qubit graph instance of Fig. \ref{UD-MIS-graphs}a for different domains up to $n_{max}=1000$ iterations. (a) and (c) depict plots of $\varepsilon(t)$ (Eq. \eqref{eq:varepsilon}) and $F(\psi^{ite}_t,\phi^{qite}_t)=|\langle\psi^{ite}_t|\phi^{qite}_t\rangle|^2$, respectively. (b) and (d) depict plots of $\Bar{\varepsilon}(t)$ and fidelity $F(\psi^{ite}_{t_{max}},\phi^{qite}_t)=|\langle\psi^{ite}_{t_{max}}|\phi^{qite}_t\rangle|^2$. The blue line shows the calculation of the fidelity with respect to the ITE state $|\psi^{ite}_{t}\rangle$.}
    \label{nqubits-6-errors-and-fidelities-vs-it}
    \end{figure}

  In Fig. \ref{nqubits-6-errors-and-fidelities-vs-it}a and \ref{nqubits-6-errors-and-fidelities-vs-it}b we plot $\varepsilon(t)$ and $\Bar{\varepsilon}(t)$, respectively. Fig \ref{nqubits-6-errors-and-fidelities-vs-it}a shows that $\ket{\psi^{ite}_t}$ and $\ket{ \phi^{qite}_t}$ start to depart from each other as the number of iterations increases, reaching a plateau for long iterations. Lower values of $\varepsilon(t)$ indicate a good fit between both states while higher values show the contrary.  If both states are almost orthogonal, this implies that $\varepsilon(t)$ is near to $\sqrt{2}$. It should be noted that for all iterations we have $\varepsilon(t) \leq \sqrt{2}$, satisfying hypothesis of Thm. \ref{th: cota_qite}.
  For a better comparison, we plot the fidelity between both states in Fig. \ref{nqubits-6-errors-and-fidelities-vs-it}c. Fig. \ref{nqubits-6-errors-and-fidelities-vs-it}d shows convergence of QITE state to $|\psi^{ite}_{t_{max}}\rangle$ as the domain size and the number of iterations increase.  
  Regarding the chosen domains, we find that $\mathcal{D}_B$ performs better than $\mathcal{D}_A$ as expected, since $\mathcal{D}_B$ is bigger than $\mathcal{D}_A$. However, this comes at the cost of needing more expectation values to obtain coefficients $a_I[l,j]$, and bigger linear systems to solve. 
%.
 From these plots we see that the low-dimensional domain $\mathcal{D}_A$  does not perform very well compared to $\mathcal{D}_B$ in matching the ITE state for large number of iterations. However, as we said in Section \ref{sec:optimization-based-on-qite}, we are interested in using QITE to get a state with eigenvalue $E_i$ such that $E_i \leq E_0 + \delta E$, with $\delta E$ a tolerable error. Moreover, we are interested in low-dimensional domains to ensure few measurements of expectation values (fewer coefficients $a[l,j]$ to be computed) and in few iterations to avoid deep quantum circuits. For this, we analyze in Fig. \ref{nqubits-6_probabilities_error_vs_iteration} and Fig. \ref{nqubits-6-proba_exacta_qite_numpy_corte_gs_y_fe_vs_shots} how $P_F^{qite}(t)$ behaves as $t$ increases together with the number of shots $M$.  
  
  \begin{figure}
    \centering
    \includegraphics[width=0.72\linewidth]{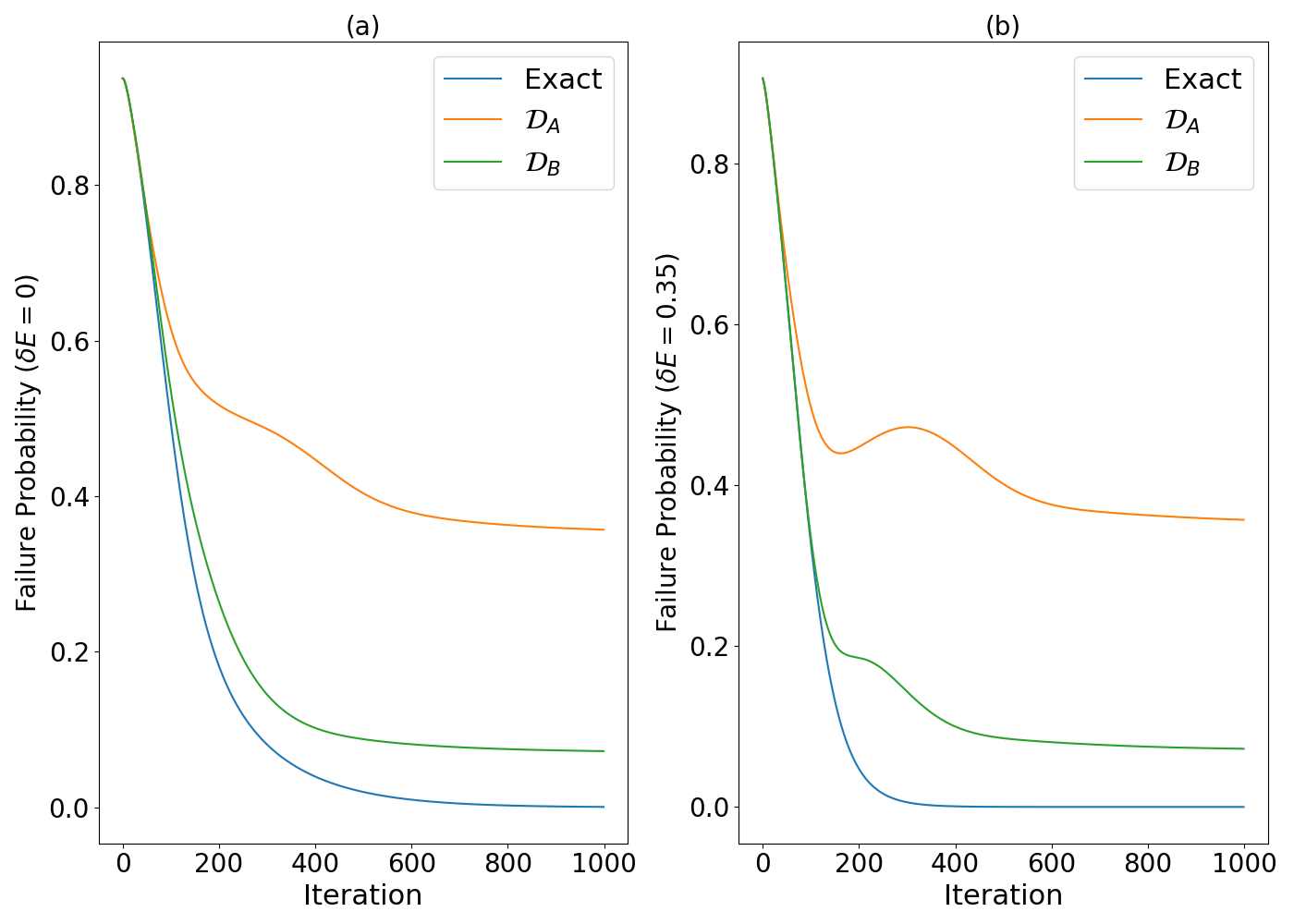}
    \caption{Failure probability for 6-qubit graph (Fig \ref{UD-MIS-graphs}a) up to $n_{max}=1000$ iterations. (a) $P_{F}^{ite}(t)$ (blue line) and $P_{F}^{qite}(t)$ for different domains and $\delta E=0$.  (b) Similar to (a) but $\delta E=0.35$. This value represents the difference between the energy of the ground state $E_0=-2$ and the energy of the first-excited state $E_1=-1.65$. The degeneracy of the ground states is $g=3$ and for first-excited states is $g=2$.}
    \label{nqubits-6_probabilities_error_vs_iteration}
\end{figure}

\begin{figure}
    \centering
    \includegraphics[width=0.72\linewidth]{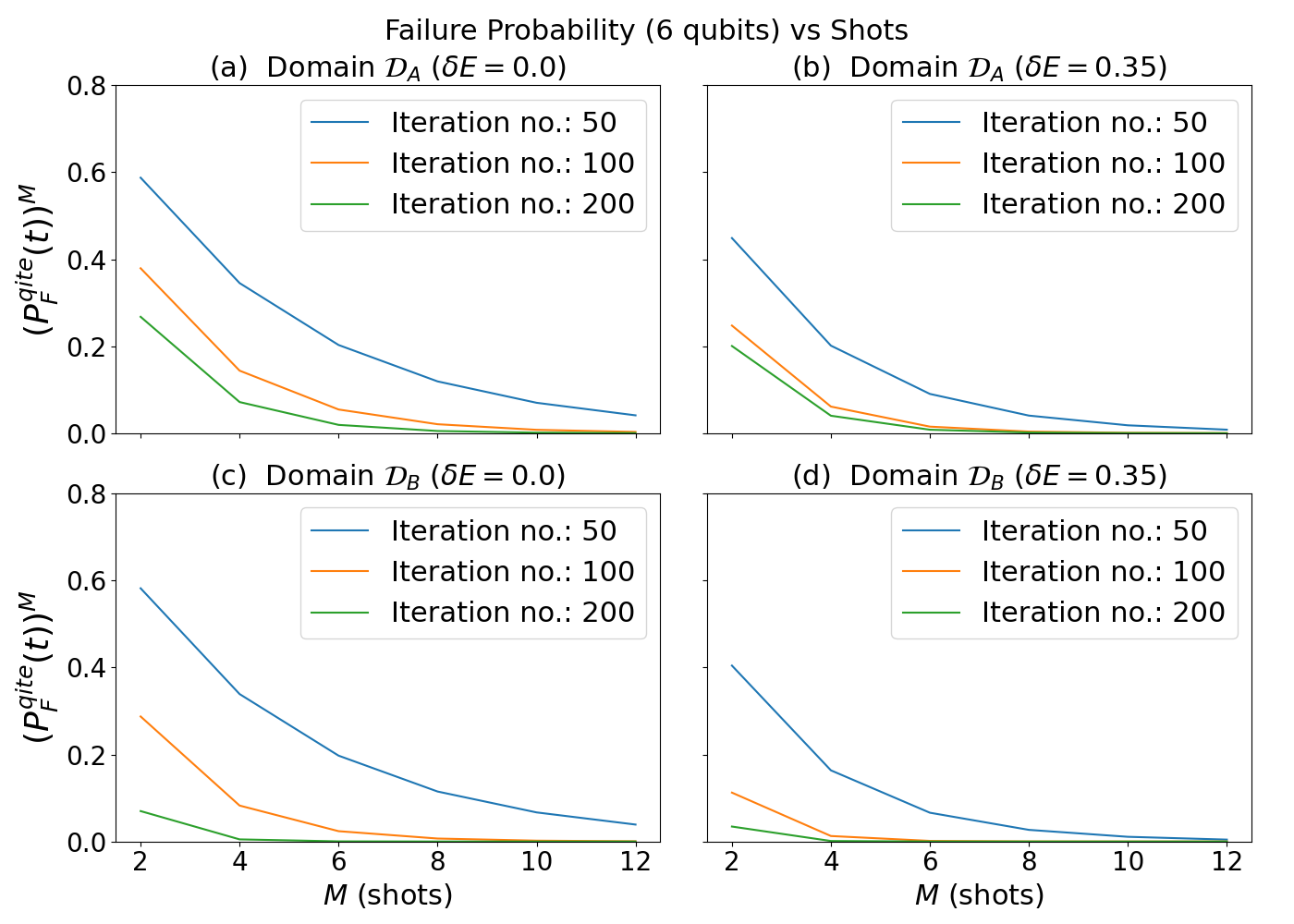}
    \caption{Failure probability $(P_{F}^{qite}(t))^M$  vs $M$ (number of shots) for 6-qubit graph for different numbers of iterations and domains.}
    \label{nqubits-6-proba_exacta_qite_numpy_corte_gs_y_fe_vs_shots}
\end{figure}

  In Fig. \ref{nqubits-6_probabilities_error_vs_iteration} we see that $P_F^{qite}(t)$ decreases a $t$ grows. Also, we see that $\mathcal{D}_B$ seems to be a much better option than $\mathcal{D}_A$ for getting low values of $P_F^{qite}(t)$. However, as stated before, even using $\mathcal{D}_A$ and a low number of iterations we might get good results, since we are only interested in getting acceptable states. This is what we see in Fig. \ref{nqubits-6-proba_exacta_qite_numpy_corte_gs_y_fe_vs_shots} when the number of shots is taken into account. $(P_F^{qite}(t))^M$ decreases considerably as the number of shots increases for different domains and number of iterations. 
  In both Fig. \ref{nqubits-6_probabilities_error_vs_iteration} and \ref{nqubits-6-proba_exacta_qite_numpy_corte_gs_y_fe_vs_shots}, $\delta E$ is chosen to include not only the ground state but also first-excited states as acceptable states (since the energy gap is roughly around $0.35$, then $\delta E = 0.35$).
  For this 6-qubit graph instance, $\delta E$ represents 17.5\% of the ground state energy (see description of Fig. \ref{nqubits-6_probabilities_error_vs_iteration}b). This implies that we are tolerating relative errors up to 17.5\%. The 
  relative error, $r$, is computed as $r=100 \times |E_0-E_i|/|E_0|$, with $E_i$ the obtained eigenvalue and $E_0$ the ground state energy. 
 Fig. \ref{nqubits-6-proba_exacta_qite_numpy_corte_gs_y_fe_vs_shots} shows, for different domains, how many iterations and number of shots are needed to get a reasonably low probability of failure for our method. It should be noted that the number of shots can be kept proportional to the number of qubits to obtain good results.

 We have also repeated the same analysis for the 8- and 10-qubit graphs up to $n_{max}=1000$ iterations. These cases have a similar behavior compared to the 6-qubit graph (see Figs. \ref{nqubits-8-errors-and-fidelities-vs-it},  \ref{nqubits-8_probabilities_error_vs_iteration}, \ref{nqubits-8-proba_exacta_qite_numpy_corte_gs_y_fe_vs_shots}
 in  Appendix \ref{app:8qubits} for 8-qubit plots and Figs. \ref{nqubits-10-errors-and-fidelities-vs-it}, \ref{nqubits-10_probabilities_error_vs_iteration}, \ref{nqubits-10-proba_exacta_qite_numpy_corte_gs_y_fe_vs_shots} 
 in \ref{app:10qubits} for 10-qubit plots, respectively).
Based on these results, for our purposes it is sufficient to test Alg. \ref{alg:cap} on randomly generated instances using $\D_A$ with $100$ iterations and $M$ at most equal to $2 N$ ($N$ number of qubits). This is what we do in the next subsection.

  \subsection{Testing random samples of UD-MIS graphs} 

\label{sec: Testing random samples}

Based on the analysis of previous subsection we have tested Alg. \ref{alg:cap} on several random UD-MIS graphs for $6$, $8$ and $10$ qubits with domain $\D_A$, $n_{max}=100$ iterations and $\tau = 0.01$. We tested on approximately $400$ graphs for $6$ qubits, $150$ graphs for $8$ qubits and $10$ graphs for $10$ qubits.  
\begin{figure}
    \centering
    \includegraphics[width=0.72\linewidth]{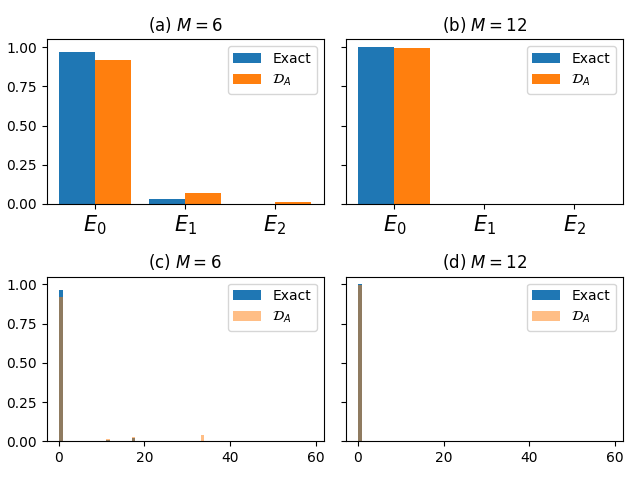}
    \caption{(a) and (b) show normalized histograms of obtained eigenvalues for 6-qubit graphs, with number of shots $M=6$ and $M=12$, respectively. (c) and (d) 
    show normalized histograms of relative errors for $M=6$ and $M=12$ shots. The number of UD-MIS instances to make these plots was around 400.}
    \label{nqubits-6-counts-vs-energy}
\end{figure}

In Fig. \ref{nqubits-6-counts-vs-energy}a and \ref{nqubits-6-counts-vs-energy}b we plot, for $M=6$ and $M=12$ shots, the  normalized histogram of the eigenvalues obtained using Alg. \ref{alg:cap} for 6-qubits graphs. 
In Fig. \ref{nqubits-6-counts-vs-energy}c and \ref{nqubits-6-counts-vs-energy}d  we plot,
the histogram of relative errors. 
From Fig. \ref{nqubits-6-counts-vs-energy}a we conclude that even for $M=6$ shots the  Alg. \ref{alg:cap} returns acceptable results, in accordance with Fig. \ref{nqubits-6-proba_exacta_qite_numpy_corte_gs_y_fe_vs_shots}.
For this random sample of 6-qubits graphs, the energy gap is roughly around $0.35$, 
representing at least $15\%$ of the ground state. 
When the gap is small, it is more difficult to find the ground state, since there is a greater probability of obtaining lower-energy excited states. However, in this case, the relative error is small. Therefore, from an  optimization point of view, all these states are expected to be acceptable states since the difference in the cost function is very small.

We have also sampled random graphs for $8$- and $10$-qubits instances and obtained similar results presented in Fig. \ref{nqubits-8-counts-vs-energy} in Appendix \ref{app:8qubits} and Fig. \ref{nqubits-10-counts-vs-energy} in \ref{app:10qubits}, respectively.
\subsection{Comparison with adiabatic quantum  computing and quantum annealing}

In this subsection we briefly compare QITE algorithm with quantum adiabatic computation and quantum annealing.
The total running time of QITE algorithm depends on the size of the chosen domain $\D$. In Motta et al \cite{Motta_2019} it is claimed that for many physical systems the domain size needed to obtain a good performance   is bounded by
$O(N_q)$, where $N_q$ is given in terms of the correlation length (see Eq. (34) of \cite{Motta_2019}). In this scenario, the total running time is of the order of $mnO(\exp(N_q))$, with $m$ the number of terms in the Hamiltonian and $n$ the number of Trotter steps. Moreover, in some situations, QITE algorithm can be used as a heuristic method choosing a  fixed domain. In this case the order of the running time can be further reduced.

Adiabatic quantum computation relies on adiabatically evolving a system from a simple initial Hamiltonian to a final one encoding the problem's solution. Success depends on maintaining the ground state throughout this slow evolution, demanding long coherence times, which are challenging to achieve.
There has been a debate about whether or not AQC would run any faster than classical algorithms due to the fact that if the problem has size $N$, one can have an evolving time exponentially increasing with $N$ in order for the system to remain in the ground state  \cite{Lucas:2013ahy, Abbas_2024}. This can be a serious disadvantage for AQC.

Quantum annealing is inspired by AQC, but relaxes some of its strict requirements. It can be considered a heuristic application of AQC. This method also involves changing from a simple initial Hamiltonian to a final one encoding the problem's solution, but it does not necessarily adhere to the adiabatic theorem as strictly. The system might not stay perfectly in the ground state, and some transitions to higher-energy states can occur. 

In order to compare the performance between quantum annealing and QITE it is necessary to define suitable metrics. We consider \emph{average maximum approximation ratio} as a metric to compare with the quantum algorithm reported in \cite{Serret_2020}. The maximum approximation ratio is defined as $\alpha= 1-r/100$, with $r$ the relative error defined in Subsection \ref{sec: Failure probability characterization}. The average maximum approximation ratio is defined as the average of the maximum approximation ratios calculated from the relative errors of histogram Fig. \ref{nqubits-6-counts-vs-energy} for 6 qubits, Fig. \ref{nqubits-8-counts-vs-energy} and Fig. \ref{nqubits-10-counts-vs-energy} for 8 and 10 qubits, respectively. The results are presented in Table \ref{tab:averagemaximumapproxratio}.

\begin{table}[h]
    \centering
    \begin{tabular}{|c|c|c|}
        \cline{2-3}
         \multicolumn{1}{c|}{} & $M=N$ & $M=2N$ \\ 
        \cline{1-3}
     6 qubits   & $0.97$ & $\sim 1.0$ \\
      8 qubits  & $0.98$  & $\sim 1.0$ \\
      10 qubits  & $0.98$ & $0.99$ \\
        \cline{1-3}
    \end{tabular}
    \caption{Average maximum approximation ratio of Alg. \ref{alg:cap} for 6, 8 and 10 qubits. $M$ is the number of shots and $N$ the number of total qubits.}
    \label{tab:averagemaximumapproxratio}
\end{table}

Fig. 4 of \cite{Serret_2020} shows the average  maximum approximation ratio for quantum annealing with Rydberg atoms for different noise levels up to 10 shots and several numbers of qubits. For less than 12 qubits it can be seen from this figure that the average maximum approximation ratio is very close to $1.0$ for $\gamma=0.0$ (noiseless case). Although our results seem to agree with theirs, we cannot conclude that Alg. \ref{alg:cap} performs similar to or different from quantum anneling of \cite{Serret_2020}. This is due to a lack of testing with a bigger number of qubits and also the analysis carried out in \cite{Serret_2020} assumes a temporal window of 2 seconds. In our case this temporal window would translate into a number of iterations, but this comparison is  complicated and we leave it for future work.

\section{Conclusions}
\label{sec:conclusions}

In this work, we applied a method based on the QITE algorithm to solve UD-MIS optimization problems. In Section \ref{sec:qite} we briefly introduced the QITE algorithm presented in \cite{Motta_2019}. 
In Section \ref{sec:optimization-based-on-qite} we described our proposed method for solving  optimization problems using QITE and we provided an upper bound for the failure probability. 

In Section \ref{sec:numeric-results} we tested
the proposed method on UD-MIS problems. First, in Subsection \ref{sec: Failure probability characterization},
we characterized the error $\varepsilon(t)$ and the failure probability for 6-, 8- and 10-qubit graphs depicted in Fig. \ref{UD-MIS-graphs}.
We obtained that $\mathcal{D}_B$ performs better than $\mathcal{D}_A$ in matching the ITE state for a large number of iterations.  However, $\D_A$ performs well enough in terms of the failure probability with a lower number of iterations and it has the advantage of involving fewer measurements of expectation values (fewer coefficients $a[l,j]$ to be computed). Moreover, we observed that the failure probability
decreases rapidly with the number of shots and the tolerable error. 
The numerical results obtained for 6-, 8- and 10-qubit graphs show that both the number of iterations and shots do not need to be drastically increased  with the number of qubits. 
Further numerical studies with more qubits will be performed 
in future works.

In Subsection \ref{sec: Testing random samples}, we tested the proposed method on approximately 400 graphs for 6 qubits,
150 graphs for 8 qubits, and 10 graphs for 10 qubits.
Our numerical findings show that for 100 iterations and a domain of type $\mathcal{D}_A$, the number of shots needed can be made proportional to the number of qubits to obtain a failure probability of less than $0.1$. It is expected that larger domains will tend to perform even better as suggested by Fig. \ref{nqubits-6-proba_exacta_qite_numpy_corte_gs_y_fe_vs_shots}, \ref{nqubits-8-proba_exacta_qite_numpy_corte_gs_y_fe_vs_shots} and \ref{nqubits-10-proba_exacta_qite_numpy_corte_gs_y_fe_vs_shots}.
We note that our chosen domains imply that the evolution given by Eq. \eqref{eq:qite-operator} produces entanglement between qubits.
This contrasts with other approaches  where QITE is applied with a linear ansatz  \cite{alam2023solvingmaxcutquantumimaginary,PhysRevA.109.052430}.

We have also plotted a distribution of relative errors in Figs. \ref{nqubits-6-counts-vs-energy}, \ref{nqubits-8-counts-vs-energy},\ref{nqubits-10-counts-vs-energy}. From this we computed the average maximum approximation ratio in Table \ref{tab:averagemaximumapproxratio}, which is an appropriate metric to compare to other proposals. In particular we compared it to results of \cite{Serret_2020}. However, these
comparisons become numerically relevant when testing against a large number of qubits and we have only explored up to 10 qubits.

Further analysis is needed to better characterize our method, including  studying its scalability with the number of qubits,  its effectiveness on real quantum computers, and its applicability to other optimization problems.
Additionally, it would be useful to investigate the effects of choosing different domains.

\begin{acknowledgments}
The authors acknowledge financial support from 
SECYT-UNC and CONICET (PIP Resol. 20204-436). We would like to thank Federico Holik for initial discussions about this work.

\end{acknowledgments}

%\vfill
\appendix

\section{Auxiliary lemmas and proofs}
\label{app:aux_proofs}

In this appendix we present the proofs of Thms. \ref{th: cota_ite} and \ref{th: cota_qite} presented in Section \ref{sec:optimization-based-on-qite}.

Proof of Theorem~\ref{th: cota_ite} (\textbf{ITE failure probability upper bound}):

\begin{proof}
\hfill

Let us consider the set of eigenvalues of $H$ that are greater than $E_0 +\delta E$,  that is, 
$\{ E_i  : E_i > E_0 + \delta E \}$. 
We call $r$ the lower index of eigenvalues of this set.

The ITE state at time $t$ can be expressed as
\begin{equation}
\ket{\psi^{ite}_t}= \frac{1}{\mathcal{K}}\sum_{i = 0}^{d-1} e^{-t E_i} | E_i\rangle, ~~~ \mbox{with} ~~~ \mathcal{K}^2  = \sum_{i = 0}^{d-1} e^{-2 t E_i},
\end{equation}
and the failure probability for this state is  
\begin{align}
    P_F^{ite }(t)=\sum_{E_i > E_0+\delta E}
        | \langle E_i \ket{\psi^{ite}_t} | ^2 = \sum_{i = r}^{d-1}
        | \langle E_i \ket{\psi^{ite}_t} | ^2  .
\end{align}

Then, 
\begin{align}
        P_F^{ite }(t) & =   \frac{\sum_{i = r}^{d-1} e^{- 2 t E_i}}{\sum_{i = 0}^{d-1}e^{-2 t E_i}}=   \frac{\sum_{i = r}^{d-1} e^{- 2 t \delta E_i}}{\sum_{i = 0}^{d-1}e^{-2 t  \delta E_i}} =    \frac{\sum_{i = r}^{d-1} e^{- 2 t \delta E_i}}{g + \sum_{i = g}^{r-1}e^{-2 t  \delta E_i} + \sum_{i = r}^{d-1}e^{-2 t  \delta E_i}},
\end{align}
with $\delta E_i =  E_i -E_0$, and we have used that the fundamental eigenvalue has degeneracy $g$. 

For $i \in [ g , r-1 ]$ we have  $\delta E_i \leq \delta E$, which implies: $ (r-g)e^{-2 t  \delta E}  = \sum_{i = g}^{r-1}e^{-2 t  \delta E} \leq \sum_{i = g}^{r-1}e^{-2 t  \delta E_i}$. Then, 
\begin{equation}
        \frac{\sum_{i = r}^{d-1} e^{- 2 t \delta E_i}}{g + \sum_{i = g}^{r-1}e^{-2 t  \delta E_i} + \sum_{i = r}^{d-1}e^{-2 t  \delta E_i}} \leq  \frac{\sum_{i = r}^{d-1} e^{- 2 t \delta E_i}}{g + (r-g) e^{-2 t  \delta E} + \sum_{i = r}^{d-1}e^{-2 t  \delta E_i}}.
\end{equation}

For $i \in [ r , d-1 ]$  we have  $\delta E_i > \delta E$, then $\sum_{i = r}^{d-1}e^{-2 t  \delta E_i} \leq  (d-r)e^{-2 t  \delta E}$. Since $f(x) = \frac{x}{A + x }$ (with $A > 0$ is a monotonically increasing function, we have 
\begin{equation}
         \frac{\sum_{i = r}^{d-1} e^{- 2 t \delta E_i}}{g + (r -g) e^{-2 t  \delta E} + \sum_{i = r}^{d-1}e^{-2 t  \delta E_i}} \leq  \frac{(d-r) e^{- 2 t \delta E}}{g + (r -g) e^{-2 t  \delta E} + (d-r)e^{-2 t  \delta E}}.
\end{equation}
Therefore, 
\begin{equation}
        P_F^{ite }(t) \leq \frac{(d-r) e^{- 2 t \delta E}}{g + (d -g) e^{-2 t  \delta E} } \leq 
        \frac{(d-g) e^{- 2 t \delta E}}{g + (d -g) e^{-2 t  \delta E} }.
\end{equation}
Finally, 
\begin{equation}
        P_F^{ite }(t) \leq 
        \frac{1}{1 + g (d -g)^{-1} e^{2 t  \delta E} }.
\end{equation}

\end{proof}

\begin{lemma} \label{lemma1}
    Let $\H$ be a finite Hilbert space and $\S_1 , \S_2 \subseteq \H$ vector subspaces such that $\S_1 \perp \S_2$ and $\S_1 \oplus \S_2 = H$. Let $\ket {\psi}  =  \ket {\psi_1} + \ket {\psi_2}$ and $\ket {\phi}  =  \ket {\phi_1} + \ket {\phi_2}$ be normalized vectors, such that 
 \begin{align}
           \ket {\psi_1} =  \cos \alpha \ket{\hat {\psi}_{1}}, ~~~ \ket {\psi_2} = \sin \alpha\ket{\hat {\psi}_{2}}, ~~~ 0 \leq \alpha \leq \pi/2,  \\
           \ket {\phi_1} =  \cos \beta \ket{\hat {\phi}_{1}}, ~~~ \ket {\phi_2} = \sin \beta \ket{\hat {\phi}_{2}},  ~~~ 0 \leq \beta\leq \pi/2.
\end{align}
with $\ket{\hat{\psi}_1} , \ket{\hat{\phi}_1} \in \S_1$ and   $\ket{\hat{\psi}_2}, \ket{\hat{\phi}_2} \in \S_2$ normalized vectors.

     Given $0 \leq  \varepsilon \leq \sqrt{2}$, if $|| \ket{\psi} -  \ket{\phi}  ||  \leq \varepsilon$, we have   $|\alpha - \beta | \leq \theta_M, $  with $\theta_{M} = 2 \arcsin(\frac{\varepsilon}{2})$.

\end{lemma}

\begin{proof}

\begin{align}
    \varepsilon^2 \geq || \ket{\psi} -\ket{\phi} ||^2 = || \ket{\psi}|| ^2 + || \ket{\phi} ||^2  -2 \mbox{Re} \langle \psi | \phi \rangle  = 2 - 2 \cos \alpha \cos \beta \, \mbox{Re} \langle \hat{\psi}_1 | \hat{\phi}_1 \rangle   - 2 \sin \alpha \sin \beta \, \mbox{Re} \langle \hat{\psi}_2 | \hat{\phi}_2 \rangle.  
\end{align}
Then,
\begin{align} \label{ineq:theta}
  1 - \frac{\varepsilon^2}{2}   \leq \cos \alpha \cos \beta \, \mbox{Re} \langle \hat{\psi}_1 | \hat{\phi}_1 \rangle    +   \sin \alpha \sin \beta \, \mbox{Re} \langle \hat{\psi}_2 | \hat{\phi}_2 \rangle   \leq \cos \alpha \cos \beta   +   \sin \alpha \sin \beta  = \cos |\alpha - \beta|.
\end{align}

From inequality \eqref{ineq:theta}, we obtain
   \begin{equation}
       |\alpha - \beta| \leq \theta_{max} = \mbox{arccos}\left( 1- \frac{\varepsilon^2}{2} \right) = 2 \, \mbox{arcsin}\left(  \frac{\varepsilon^2}{2}\right).
   \end{equation}

\end{proof}

Proof of Theorem~\ref{th: cota_qite} (\textbf{QITE failure probability upper bound}):

\begin{proof}

Let us consider the set of eigenvalues of $H$ that are greater than $E_0 +\delta E$,  that is, 
$\{ E_i  : E_i > E_0 + \delta E \}$, We call $r$ the lower index of eigenvalues of this set, and we define two sets of eigenvectors of $H$, $S_1 =  \{ \ket{E_i} \}_{i=0}^{r-1}$ and $S_2 =  \{ \ket{E_i} \}_{i=r}^{d-1}$. Then, we consider the linear spanned subspaces $\S_1 = \mbox{span}(S_1)$ and $\S_2 = \mbox{span}(S_2)$. This vector subspaces satisfy $\S_1 \perp \S_2$ and $\S_1 \oplus \S_2 = H$.

We can express $\ket {\psi^{ite}_t}  =  \ket {\psi_1} + \ket {\psi_2}$ and $\ket {\phi^{ite}_t}  =  \ket {\phi_1} + \ket {\phi_2}$, with $\ket{\psi_1} , \ket{\phi_1} \in \S_1$ and   $\ket{\psi_2}, \ket{\phi_2} \in \S_2$. Moreover, since $ \ket {\psi^{ite}_t}$ and $ \ket {\phi^{qite}_t}$ are normalized vectors, we can express them as follows
\begin{align}
           \ket {\psi_1} &=  \cos \alpha \ket{\hat {\psi}_{1}}, ~~~ \ket {\psi_2} = \sin \alpha\ket{\hat {\psi}_{2}}, ~~~ 0 \leq \alpha \leq \pi/2,  \\
           \ket {\phi_1} &=  \cos \beta \ket{\hat {\phi}_{1}}, ~~~ \ket {\phi_2} = \sin \beta\ket{\hat {\phi}_{2}},  ~~~ 0 \leq \beta\leq \pi/2.
\end{align}
with $\ket{\hat{\psi}_1} , \ket{\hat{\phi}_1} \in \S_1$ and   $\ket{\hat{\psi}_2}, \ket{\hat{\phi}_2} \in \S_2$ normalized vectors.

 Given $0 \leq  \varepsilon \leq \sqrt{2}$, if $||\ket {\psi^{ite}_t}  - \ket {\phi^{qite}_t} || \leq \varepsilon$, due to Lemma \ref{lemma1}, we have   $|\alpha - \beta | \leq \theta_M, $  with $\theta_{M} = 2 \arcsin(\frac{\varepsilon}{2})$.

Moreover, the failure probabilities $P_F^{ite}(t)$ and $P_F^{qite}(t)$ are given by 
\begin{align}
\label{eq:Pf_ite_angulo}
P_F^{ite}(t) &= || \ket{\psi_2} ||^2 = \sin^2 \alpha,  \\
\label{eq:Pf_qite_angulo}
P_F^{qite}(t) &=  || \ket{\phi_2} ||^2 = \sin^2 \beta.
\end{align}

Then,  we have  
\begin{align}
|P_F^{qite}(t) - P_F^{ite}(t)| &= |\sin ^2\beta  -\sin ^2\alpha|  = |\sin (\alpha + \beta) \sin(\alpha -\beta)| \leq \sin|\alpha -\beta|,
\end{align}
where in the first step we use Eqs. \eqref{eq:Pf_ite_angulo} and \eqref{eq:Pf_qite_angulo}, in the second step we use a trigonometric identity, and in the last step we have used $|\sin (\alpha + \beta)| \leq 1$. 

Since $|\alpha - \beta | \leq \theta_M \leq \frac{\pi}{2}$, we have $\sin|\alpha -\beta| \leq \sin \theta_M $. Therefore, 
\begin{align}
|P_F^{qite}(t) - P_F^{ite}(t)| \leq \sin \theta_M.
\end{align}
Finally, using  $\sin (2 \arcsin x)= 2x \sqrt{1 - x^2}$ and $\theta_{M} = 2 \arcsin(\frac{\varepsilon}{2})$, we obtain
\begin{align}
    |P_F^{qite}(t) - P_F^{ite}(t)|  \leq  \varepsilon \sqrt{1 - \frac{\varepsilon ^2}{4}}.
\end{align}

\end{proof}

\section{Some remarks and extra results of Section \ref{sec:numeric-results}}\label{app:numerical}
In this Appendix we briefly summarize the details of the numerical analysis.
We have discretized time as $t=\tau n$ with $\tau=0.01$. When analyzing single instances of graphs we simulated up to 1000 iterations (i.e. $t_{max}=10$) and when testing several random graphs we set $n_{max}=100$. A central aspect of QITE is the chosen domain $\mathcal{D}$. Since UD-MIS graphs in general do not have a regular lattice shape, it is difficult to propose natural domains for them. Nevertheless, we use the following recipe for the Hamiltonian \eqref{eq:UD-MIS_Hamiltonian}. Single $Z_i$ terms are kept separated from interaction terms $Z_iZ_{i'}$. And for each single $Z_i$ the associated $D_{l,j}$  equals exactly the support of qubit $i$ (for all $1\leq j \leq n_{max}$). For the interaction terms we define two different domains $\mathcal{D}_A$ and $\mathcal{D}_B$. $\mathcal{D}_A$ assigns  each interaction term to a domain of exactly the same support on which the interaction term acts on. That is, $Z_iZ_{i'}$ acts non-trivially on qubits $i$ and $i'$, so the associated domain has non-trivial support only on $i$ and $i'$. On the other hand, $\mathcal{D}_B$ is similar to  $\mathcal{D}_A$ except that the domain which is associated with $Z_iZ_{i'}$ is expanded to a total of 4 qubits containing not only $i$ and $i'$ but other two qubits linked to $i$ and $i'$. In this case, if $i$ and $i'$ are linked with several nodes, we  randomly choose any of those nodes (since every link in a UD-MIS graphs weighs the same). Note that Hamiltonian (\ref{eq:UD-MIS_Hamiltonian}) does not depend on the distances between graph nodes, otherwise we would expect this feature to provide a natural way to build domains.

The complete domain $\mathcal{D}_A$ for the 6-qubit instance of Fig. \ref{UD-MIS-graphs}a is: 
\begin{equation}\label{eq:D_A-6q}
\begin{aligned}
    \mathcal{D}_A & =[(0),(1),(2),(3),(4),(5),\\
    &~~~~(0,1),(0,3),(0,5), \\
    &~~~~(1,2),(1,3),(1,4),(1,5),\\
    &~~~~(2,3),(2,4),\\
    &~~~~(3,4),(3,5),\\
    &~~~~(4,5)].
\end{aligned}
\end{equation}
The first line of Eq. \eqref{eq:D_A-6q} shows domains which are associated with each single $Z_i$ term of the Hamiltonian \eqref{eq:UD-MIS_Hamiltonian} and the rest of the lines show domains which are associated with interaction terms $Z_iZ_{i'}$. It should be understood that each domain represents a complete Pauli basis for the qubits involved. The chosen $\mathcal{D}_B$ for the 6-qubit instance of Fig. \ref{UD-MIS-graphs}a is:
\begin{equation}\label{eq:D_B-6q}
\begin{aligned}
    \mathcal{D}_B & =[(0),(1),(2),(3),(4),(5),\\
    &~~~~(0,1,3,5),(0,1,2,3),(0,1,4,5), \\
    &~~~~(1,2,4,5),(0,1,3,4),(1,2,4,5),(0,1,3,5),\\
    &~~~~(0,2,3,4),(2,3,4,5),\\
    &~~~~(0,3,4,5),(1,3,4,5),\\
    &~~~~(0,2,4,5)].
\end{aligned}
\end{equation}
We say the size of $\mathcal{D}_B$ is bigger than the size of $\mathcal{D}_A$ since the first one contains  bigger qubit supports for interaction terms. This implies  measuring more expectation values per Trotter sub-step.
\subsection{Results for 8-qubits graphs}\label{app:8qubits}

In this section we present the results for random graphs of 8 qubits.
As in the case of 6-qubit graphs, we obtained that $\mathcal{D}_B$ performs better than $\mathcal{D}_A$ in matching the ITE state for large number of iterations. However, $\D_A$ performs well enough in terms of the failure probability. Again, we obtained that the failure probability
decreases when increasing the number of iterations, the number of shots, and the tolerable error. As in the case of 6 qubits, the number of iterations and shots do not need to be increased exponentially with the number of qubits.

\begin{figure}[H]
        \centering
        \includegraphics[width=0.70\textwidth]{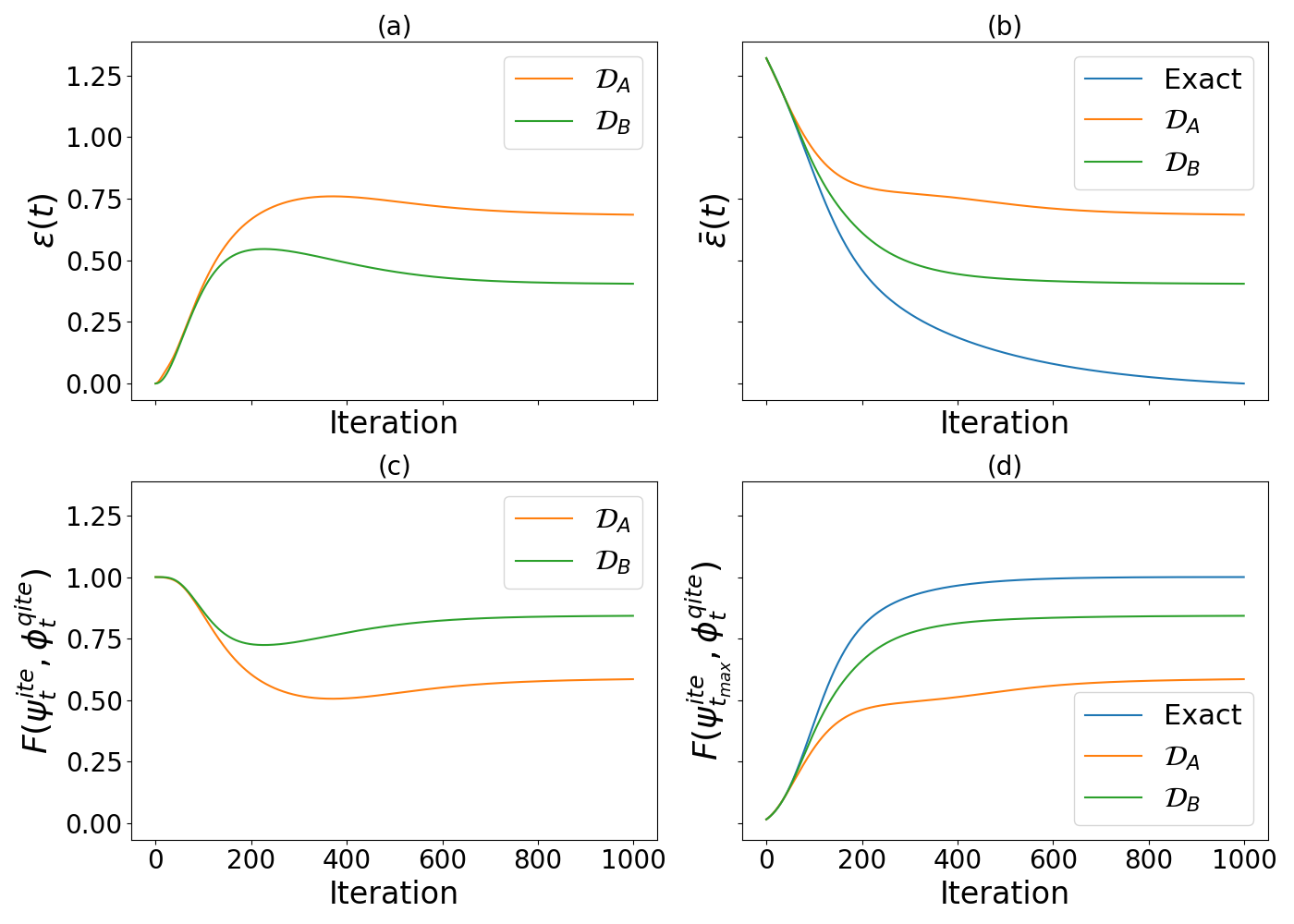} 
    \caption{Error and fidelity results for the 8-qubit graph instance of Fig. \ref{UD-MIS-graphs}b for different domains up to $n_{max}=1000$ iterations. (a) and (c) depict plots of $\varepsilon(t)$ (Eq. (\ref{eq:varepsilon})) and $F(\psi^{ite}_t,\phi^{qite}_t)=|\langle\psi^{ite}_t|\phi^{qite}_t\rangle|^2$, respectively. (b) and (d) depict plots of $\Bar{\varepsilon}(t)$ (Eq. (\ref{eq:Barepsilon})) and fidelity $F(\psi^{ite}_{t_{max}},\phi^{qite}_t)=|\langle\psi^{ite}_{t_{max}}|\phi^{qite}_t\rangle|^2$. The blue line corresponds to calculate the fidelity with respect to the ITE state $|\psi^{ite}_{t}\rangle$. }
    \label{nqubits-8-errors-and-fidelities-vs-it}
    \end{figure}

\begin{figure}[H]
    \centering
    \includegraphics[width=0.70\linewidth]{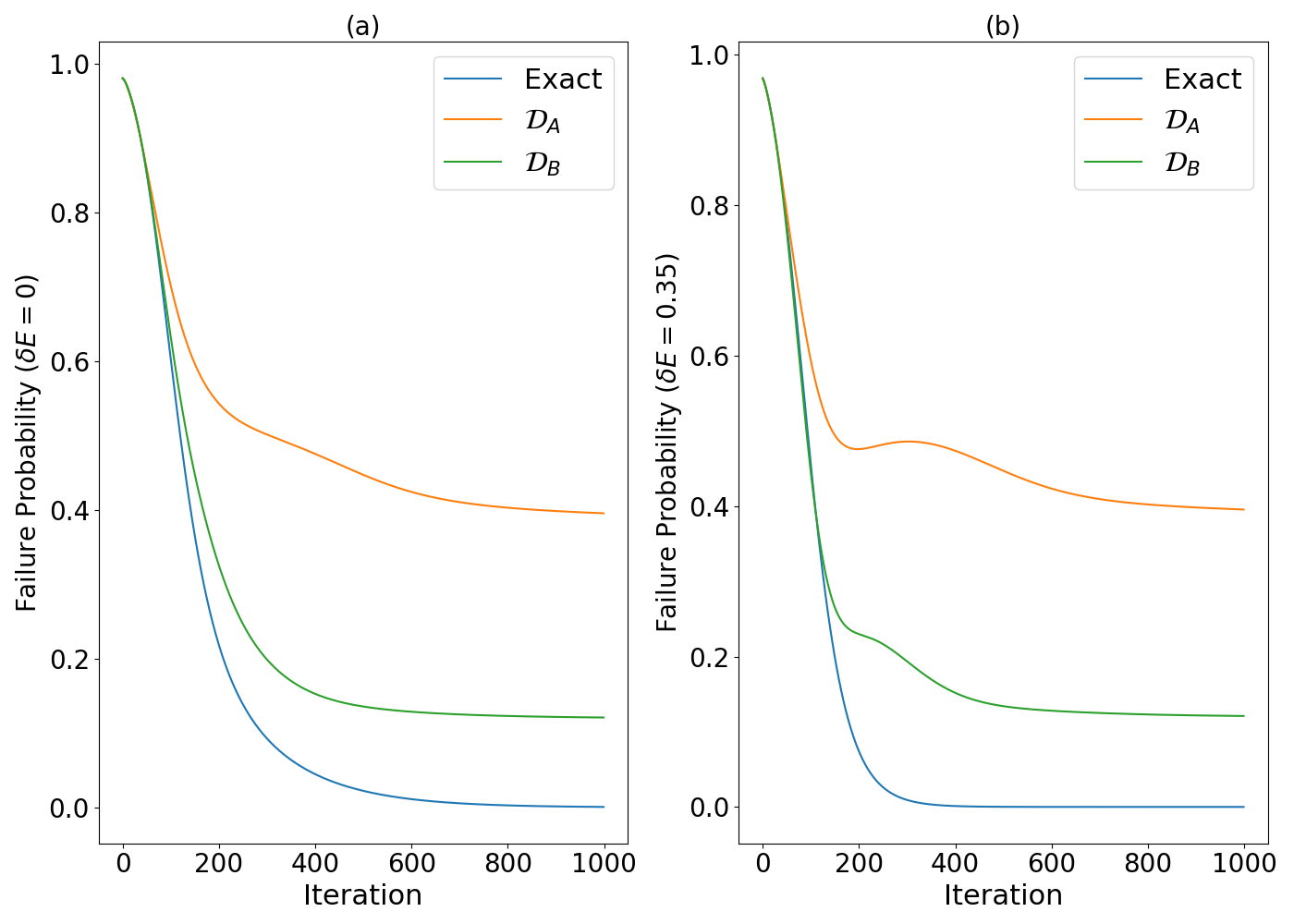}
    \caption{Failure probability for 8-qubit graph (Fig \ref{UD-MIS-graphs}b) up to $n_{max}=1000$ iterations. (a) $P_{F}^{ite}(t)$ (blue line) and $P_{F}^{qite}(t)$ for different domains and $\delta E=0$.  (b) Similar to (a) but $\delta E=0.35$. This value represents the difference between the energy ground-state $E_0=-3$ and the energy of the first-excited state $E_1=-2.65$. The degeneracy of the ground states is $g=4$ and for first-excited states is $g=3$.}

    \label{nqubits-8_probabilities_error_vs_iteration}
\end{figure}

\begin{figure}[H]
    \centering
    \includegraphics[width=0.70\linewidth]{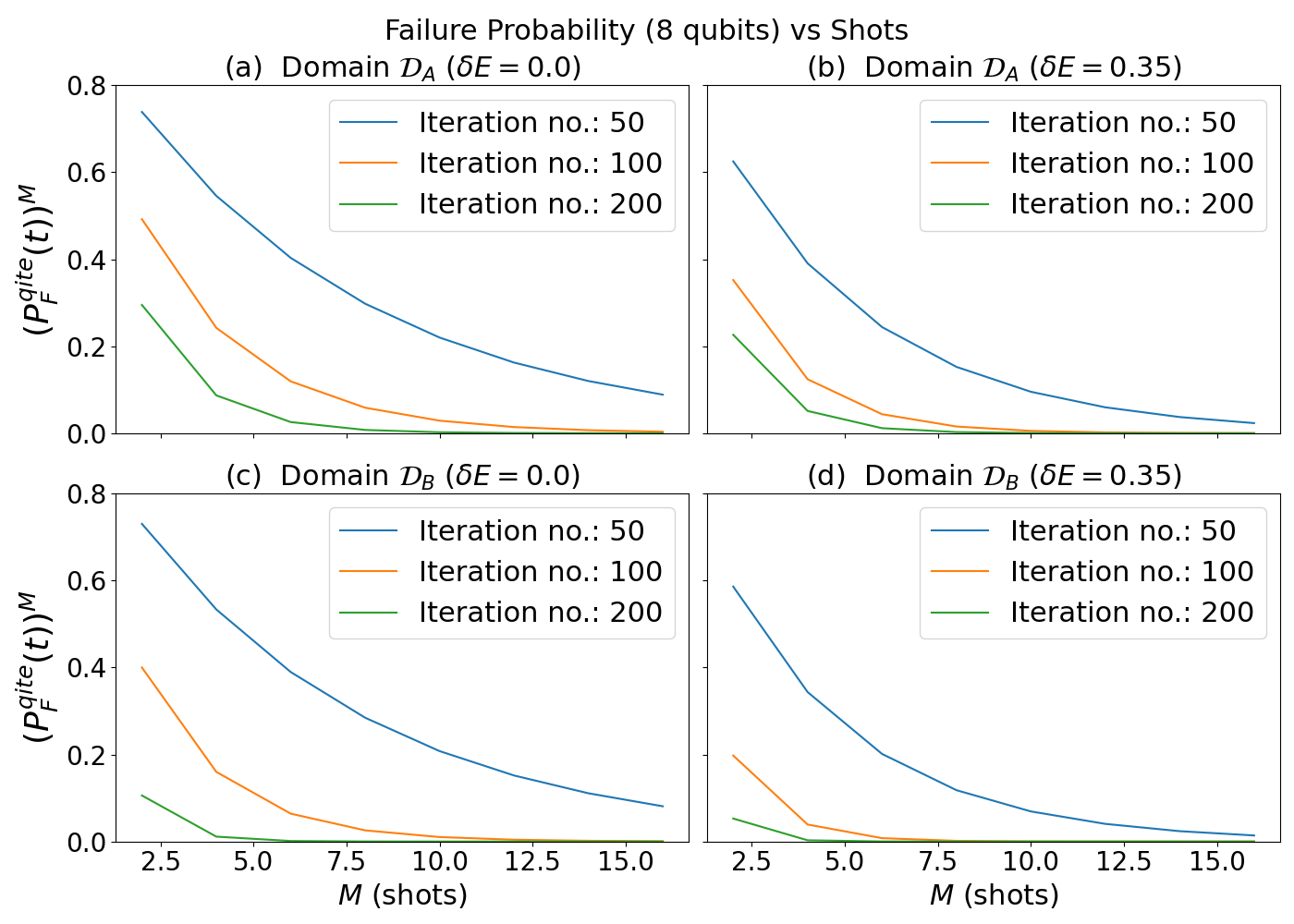}

    \caption{Failure probability $(P_{F}^{qite}(t))^M$  vs $M$ (number of shots) for 8-qubit graph for different iterations and domains.}
    \label{nqubits-8-proba_exacta_qite_numpy_corte_gs_y_fe_vs_shots}
\end{figure}

\begin{figure}[H]
    \centering
    \includegraphics[width=0.70\linewidth]{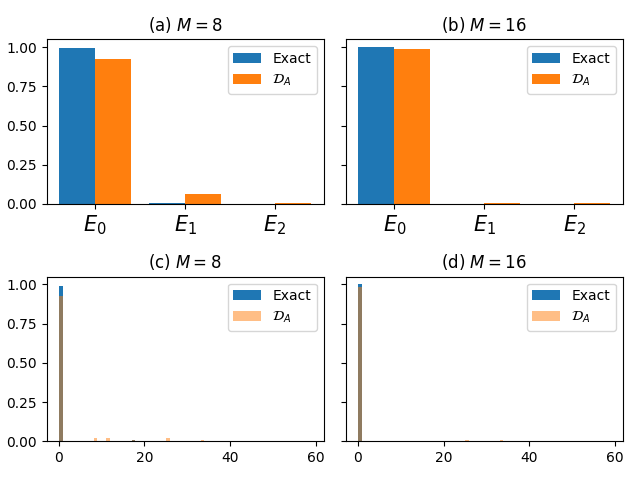}
    \caption{
    (a) and (b) show normalized histograms of obtained eigenvalues for 8-qubit graphs, with number of shots $M=8$ and $M=16$, respectively. (c) and (d) 
    show normalized histograms of relative errors for $M=8$ and $M=16$ shots. The number of UD-MIS instances to make these plots was around 150.}
    \label{nqubits-8-counts-vs-energy}
\end{figure}

\subsection{Results for 10-qubits graphs}\label{app:10qubits}

In this section we 
show the results for random graphs of 10 qubits. In this case we only consider 10 random graphs due to the numerical complexity of the calculations.
As in the other cases, $\D_A$ performs well enough in terms of the failure probability, but 
$\mathcal{D}_B$ performs better than $\mathcal{D}_A$ in matching the ITE state for large number of iterations. Moreover, it is observed that the failure probability
decreases when increasing the number of iterations, the number of shots, and the tolerable error. The number of iterations and shots do not need to be increased exponentially with the number of qubits.

\begin{figure}[H]
        \centering
        \includegraphics[width=0.70\textwidth]{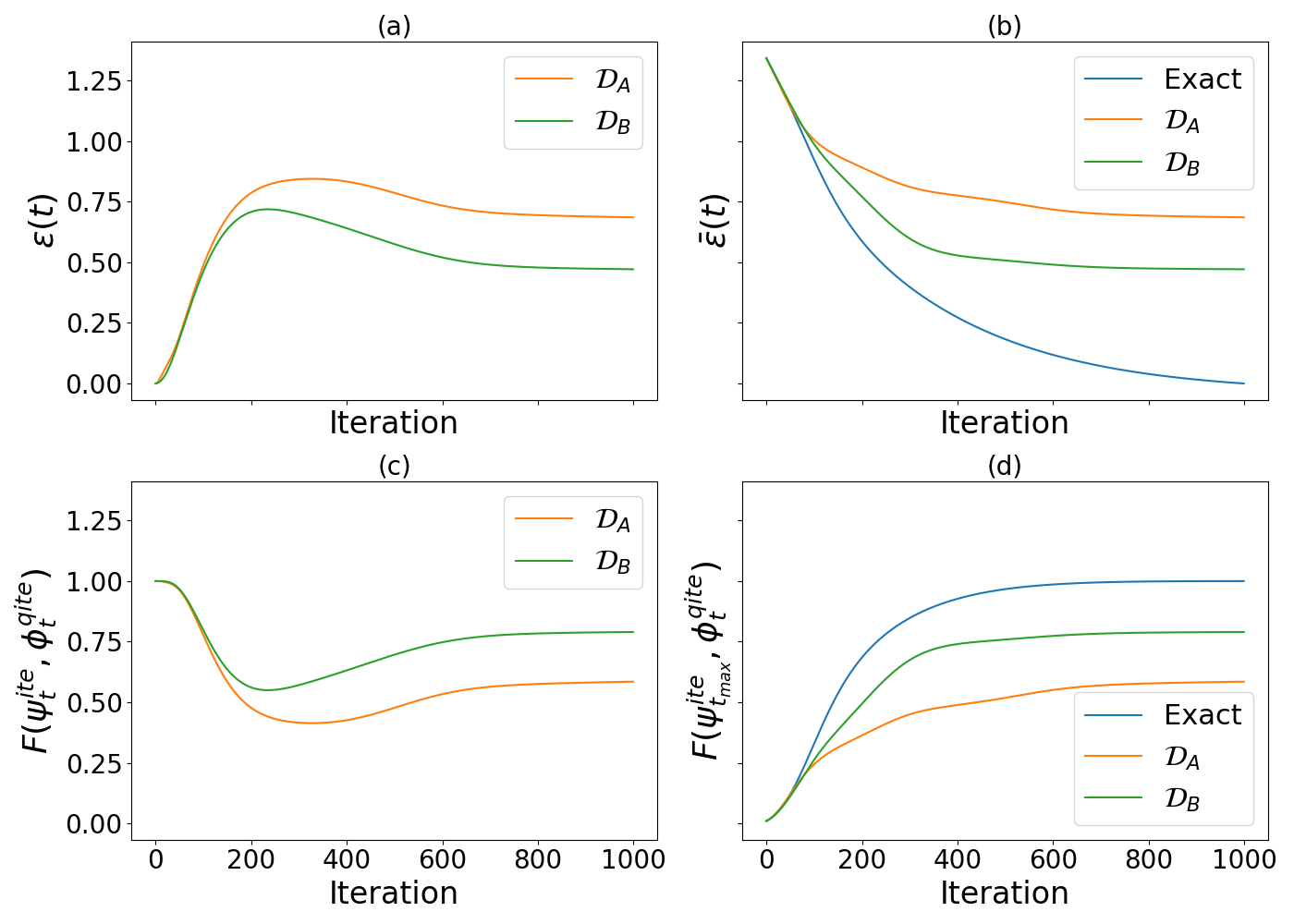} 
    \caption{Error and fidelity results for the 10-qubit graph instance of Fig. \ref{UD-MIS-graphs}c for different domains up to $n_{max}=1000$ iterations. (a) and (c) depict plots of $\varepsilon(t)$ (Eq. (\ref{eq:varepsilon})) and $F(\psi^{ite}_t,\phi^{qite}_t)=|\langle\psi^{ite}_t|\phi^{qite}_t\rangle|^2$ respectively. (b) and (d) depict plots of $\Bar{\varepsilon}(t)$ (Eq. \eqref{eq:Barepsilon}) and fidelity $F(\psi^{ite}_{t_{max}},\phi^{qite}_t)=|\langle\psi^{ite}_{t_{max}}|\phi^{qite}_t\rangle|^2$. The blue line corresponds to calculate the fidelity with respect to the ITE state $|\psi^{ite}_{t}\rangle$.
        }
    \label{nqubits-10-errors-and-fidelities-vs-it}
    \end{figure}

\begin{figure}[H]
    \centering
    \includegraphics[width=0.70\linewidth]{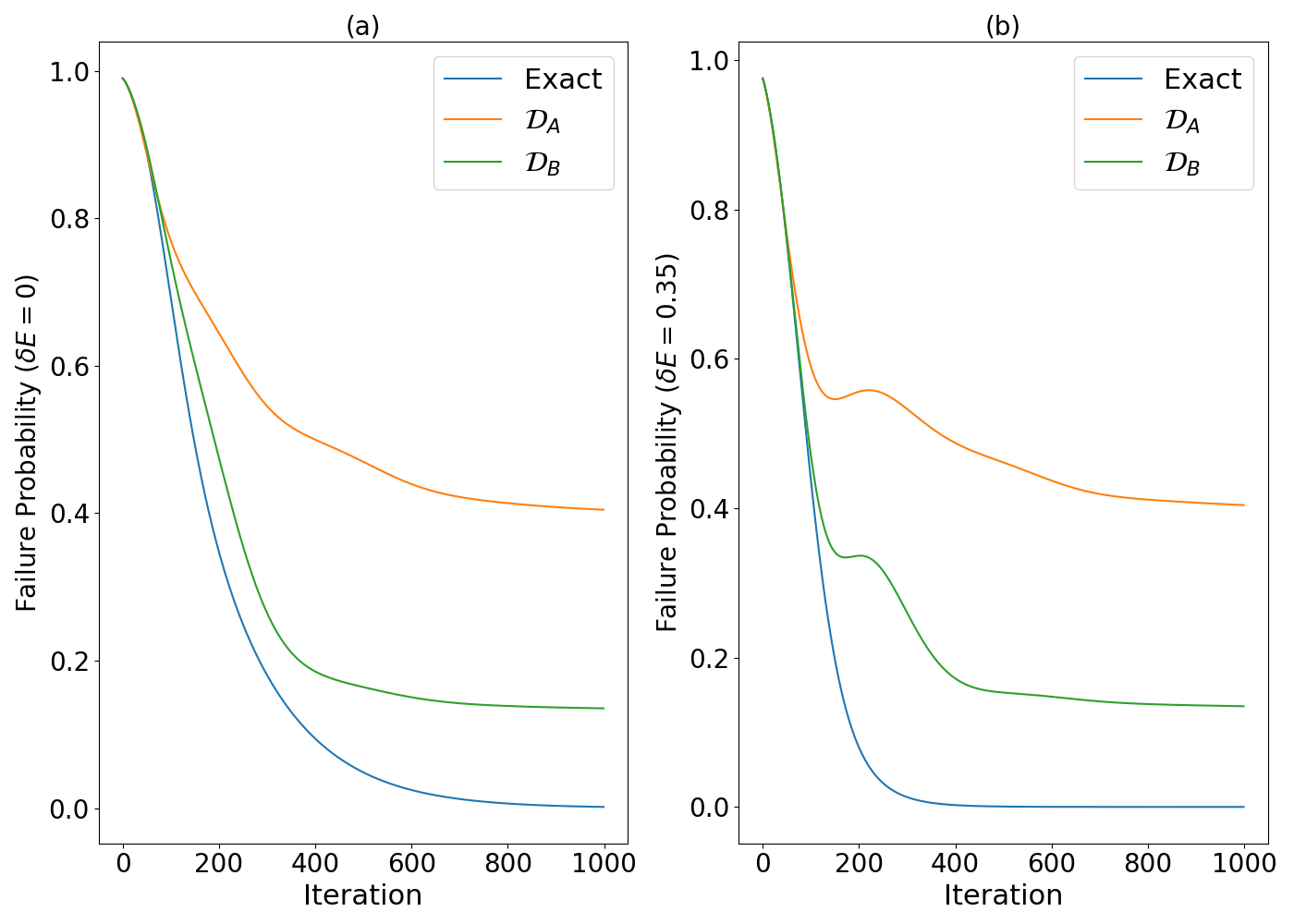}
   \caption{Failure probability for 10-qubit graph (Fig \ref{UD-MIS-graphs}c) up to $n_{max}=1000$ iterations. (a) $P_{F}^{ite}(t)$ (blue line) and $P_{F}^{qite}(t)$ for different domains and $\delta E=0$.  (b) Similar to (a) but $\delta E=0.35$. This value represents the difference between the energy ground-state $E_0=-3$ and the energy of the first-excited state $E_1=-2.65$. The degeneracy of the ground states is $g=9$ and for first-excited states is $g=15$.}
    \label{nqubits-10_probabilities_error_vs_iteration}
\end{figure}

\begin{figure}[H]
    \centering
    \includegraphics[width=0.70\linewidth]{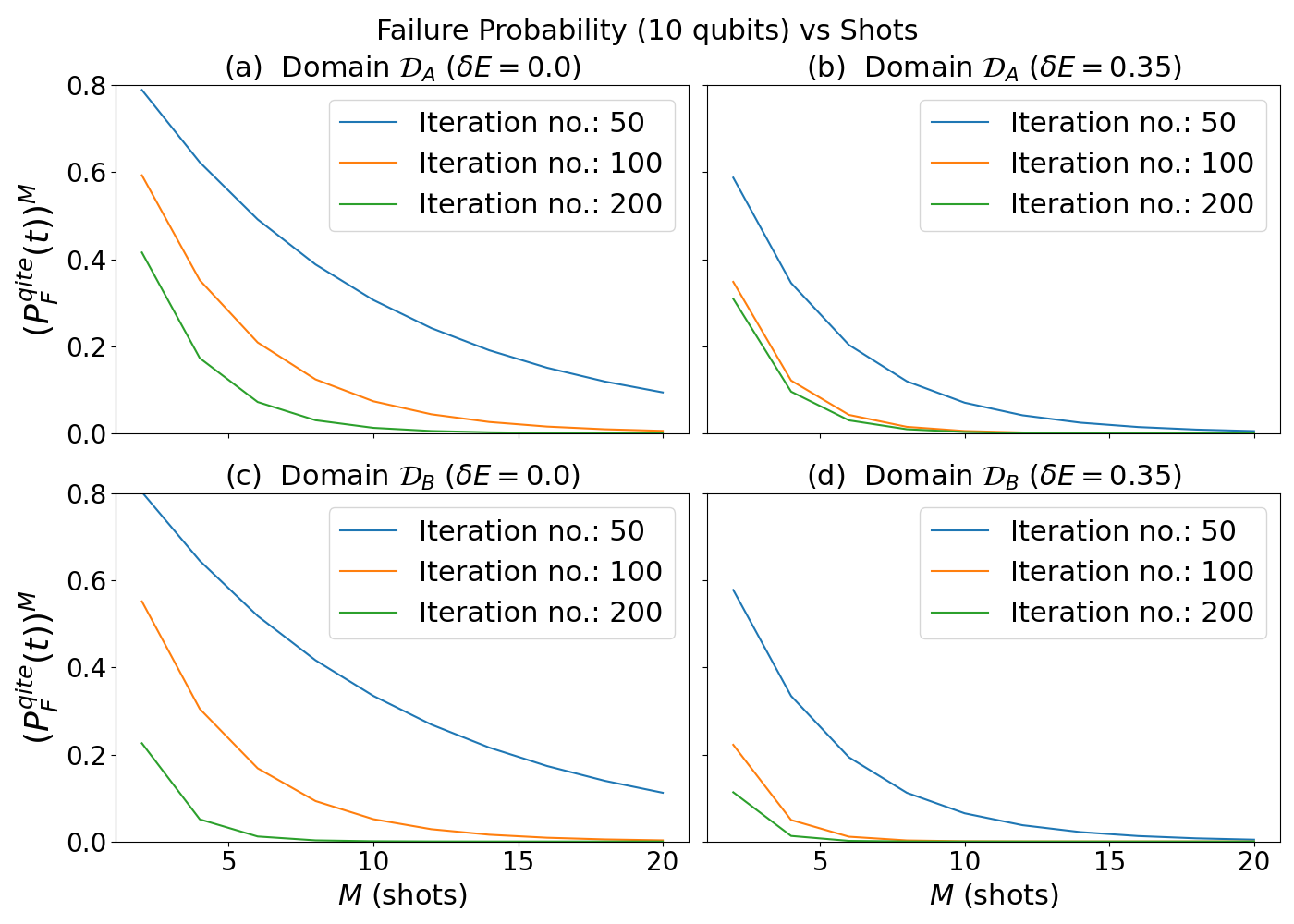}
    \caption{Failure probability $(P_{F}^{qite}(t))^M$  vs $M$ (number of shots) for 10-qubit graph for different number of iterations and domains.}
    \label{nqubits-10-proba_exacta_qite_numpy_corte_gs_y_fe_vs_shots}
\end{figure}

\begin{figure}[H]
    \centering
    \includegraphics[width=0.70\linewidth]{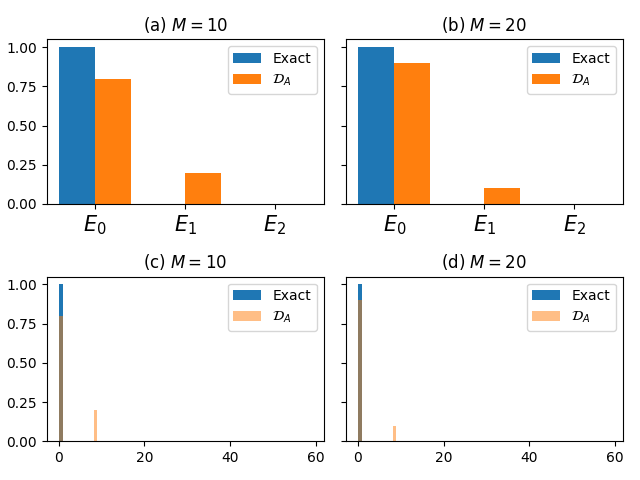}
    \caption{(a) and (b) show normalized histograms of obtained eigenvalues for 6-qubit graphs, with number of shots $M=10$ and $M=20$, respectively. (c) and (d) 
    show normalized histograms of relative errors for $M=10$ and $M=20$ shots. The number of UD-MIS instances to make these plots was 10.}
    \label{nqubits-10-counts-vs-energy}
\end{figure}

\bibliographystyle{ieeetr}
\bibliography{Referencias}

\end{document}